\documentclass[fconference,a4paper]{IEEEtran}

\usepackage{amssymb,amsfonts,amsmath,amstext,mathrsfs}
\usepackage{latexsym}
\usepackage{epsfig,color,graphics,epsf}
\usepackage{enumerate,cite}
\usepackage{bbm}

\usepackage{amsmath}
\usepackage{amssymb}
\usepackage{amsfonts}
\usepackage{amsthm}
\usepackage{graphicx}
\usepackage{float}
\usepackage{graphics}
\usepackage{graphicx}
\usepackage{euscript}

\DeclareMathOperator*{\argmax}{arg\,max}

\newtheorem{theorem}{Theorem}
\newtheorem{lemma}{Lemma}

\newtheorem{proposition}{Proposition}

\newtheorem{corollary}{Corollary}
\newtheorem{definition}{Definition}

\newcommand{\EE}{\mathbb{E}}

\newcommand{\dotleq}{%
\DOTSB\mathrel{\mathop{\kern0pt \leq}\limits^{\textstyle.}}}
\newcommand{\dotgeq}{%
\DOTSB\mathrel{\mathop{\kern0pt \geq}\limits^{\textstyle.}}}

\newcommand{\reals} {\mathbb{R}}

\newcommand{\beq} {\begin{equation}}
\newcommand{\eeq} {\end{equation}}
\newcommand{\beqa} {\begin{align}}
\newcommand{\eeqa} {\end{align}}

\newcommand{\indicator}{\mathbbm{1}}
\newcommand{\e}{\mathrm{e}}
\newcommand{\vol}{\mathrm{vol}}

\newcommand {\bx} {\boldsymbol{x}}
\newcommand {\by} {\boldsymbol{y}}

\newcommand {\bX} {\boldsymbol{X}}
\newcommand {\bY} {\boldsymbol{Y}}

\newcommand{\calC}{{\mathcal C}}
\newcommand{\calD}{{\mathcal D}}
\newcommand{\calE}{{\mathcal E}}

\newcommand{\calI}{{\mathcal I}}

\newcommand{\calP}{{\mathcal P}}

\newcommand{\calT}{{\mathcal T}}

\newcommand{\calV}{{\mathcal V}}

\newcommand{\calX}{{\mathcal X}}
\newcommand{\calY}{{\mathcal Y}}

\allowdisplaybreaks
\IEEEoverridecommandlockouts

\begin{document}


\title{Generalized Random Gilbert-Varshamov Codes}

\author{Anelia Somekh-Baruch, Jonathan Scarlett, and Albert Guill\'en i F\`abregas
\thanks{A.~Somekh-Baruch is with the Faculty of Engineering, Bar-Ilan University, Ramat Gan 52900, Israel (e-mail: somekha@biu.ac.il).

J.~Scarlett is with the Department of Computer Science and the Department of Mathematics, National University of Singapore, Singapore (e-mail: scarlett@comp.nus.edu.sg).

A.~Guill\'en i F\`abregas is with the Department of Information and Communication
Technologies, Universitat Pompeu Fabra, Barcelona 08018, Spain,
also with the Instituci\'o Catalana de Recerca i Estudis Avan\c{c}ats (ICREA),
Barcelona 08010, Spain, and also with the Department of Engineering, University
of Cambridge, Cambridge CB2 1PZ, U.K. (e-mail: guillen@ieee.org).

This work was supported in part by the Israel Science Foundation under grant 631/17, 
by the European Research Council under Grant 725411, by the Spanish Ministry of Economy and Competitiveness under Grant TEC2016-78434-C3-1-R, and by an NUS Early Career Research Award.
This paper was presented in part at the 2018 International Zurich Seminar, the 2018 Conference on Information Sciences and Systems, Princeton University and the 2018 IEEE International Symposium on Information Theory. 
}

}

\date{\today}

\maketitle
\begin{abstract}

We introduce a random coding technique for transmission over discrete memoryless channels, reminiscent of the basic construction attaining the Gilbert-Varshamov bound for codes in Hamming spaces. The code construction is based on drawing codewords recursively from a fixed type class, in such a way that a newly generated codeword must be at a certain minimum distance from all previously chosen codewords, according to some generic distance function. We derive an achievable error exponent for this construction, and prove its tightness with respect to the ensemble average. We show that the exponent recovers the Csisz\'{a}r and K{\"o}rner exponent as a special case, which is known to be at least as high as both the random-coding and expurgated exponents, and we establish the optimality of certain choices of the distance function. In addition, for additive distances and decoding metrics, we present an equivalent dual expression,  along with a generalization to infinite alphabets via cost-constrained random coding.

\end{abstract}

\allowdisplaybreaks[0]

\section{Introduction}

The problem of characterizing the error exponents of channel coding has been studied extensively since the early days of information theory.  The goal is to establish bounds on the rate of decay of the error probability for fixed rates below capacity.  While the random coding exponent and sphere packing exponent establish the exact error exponent at rates sufficiently close to capacity, the optimal exponent at low rates has generally remained open, except in the limit of zero rate.

For discrete memoryless channels (DMC), improvements over the random-coding exponent at low rates are provided by the expurgated exponent. The idea of the original derivation of this exponent is simple \cite{Gallager68}: After generating the codewords independently at random, remove a fraction of the worst codewords (i.e., those with the highest error probability) while keeping enough so that the loss in the rate is negligible.  Alternative derivations have since appeared based on the method of types and random selection \cite{CsiszarKorner81}, graph decomposition techniques \cite{CsiszarKorner81graph}, and type class enumeration \cite{ScarlettPengMerhavMartinezGuilleniFabregas_mismatch_2014_IT}. For other related works, see \cite{Merhav_List_Decoding_IT_2014}, \cite{Merhav_Generalized_2017}, \cite{SomekhBaruch_mismatchachievableIT2014} and references therein. 

In the literature, many of the most commonly-studied error exponents admit (at least) two equivalent forms:
\begin{itemize}
    \item A {\em primal} expression is written as a minimization over joint distributions subject to suitable constraints, and is typically derived using the method of types \cite{CsiszarKorner81}.  Such derivations often have the advantage of immediately proving tightness with respect to the random-coding ensemble under consideration.
    \item A {\em dual} expression is written as a maximization over auxiliary parameters, and is typically derived using Gallager-type techniques \cite{Gallager68} such as Markov's inequality and $\min\{1,\alpha\}\leq \alpha^\rho$ for $\rho\in[0,1]$.  Such derivations often have the advantage of extending to continuous-alphabet memoryless channels. In addition, dual expressions provide achievable exponents for arbitrary {\em fixed} choices of the auxiliary parameters. 
\end{itemize}
This naming convention arises from the fact that the equivalence of the expressions is proved using Lagrange duality.  In the setting of the present paper with a general additive decoding metric, such equivalences were given for achievable rates in \cite{MerhavKaplanLapidothShamai94}, for random coding error exponents in \cite{ScarlettMartinezGuilleniFabregas2012AllertonSU,ScarlettMartinezGuilleniFabregas_mismatch_2014_IT}, and for expurgated exponents in \cite{ScarlettPengMerhavMartinezGuilleniFabregas_mismatch_2014_IT}.

In this paper, we introduce a recursive random coding construction that achieves the exponent of Csisz\'{a}r and K{\"o}rner \cite{CsiszarKorner81graph}, thus achieving the maximum of the random-coding and expurgated exponents.  The code construction is based on drawing codewords recursively from a fixed type class, in such a way that a newly generated codeword must be at a certain minimum distance from all previously chosen codewords, according to some generic distance function. This construction is reminiscent to those in the binary Hamming space dating back to the 1950s \cite{gilbert1952comparison,varshamov1957estimate,siforov1956noise} (see also \cite{levenshtein1960class,brualdi1993greedy,conway1986lexicographic,trachtenberg2000error}), known to achieve the Gilbert-Varshamov bound.  We therefore adopt the name generalized {\em Random Gilbert-Varshamov (RGV) code} for our randomized construction with a general distance function and constant-composition codewords. 
A related work by 
Blahut \cite{BlahutComposition} studied properties of the Bhattacharyya and the equivocation distance functions and derived generalized bounds similar to the Gilbert-Varshamov and Elias bounds. These bounds are used to derive an upper bound on the reliability function. Connections with the expurgated exponent are also explored.
Another related work is that of Barg and Forney  \cite{BargForney2002}, who showed that for the binary symmetric channel (BSC), typical linear codes, whose minimum distance attains the Gilbert-Varshamov bound, achieve the expurgated exponent.

\subsection{Contributions}

The main contributions of this work are as follows:
\begin{itemize}
    \item As outlined above, we introduce the generalized RGV construction, and analyze its error exponent for a given DMC, decoding metric, and distance function. 
    Similarly to the Gilbert-Varshamov bound, our construction induces a tradeoff between the rate and the minimum distance of the code. 
  As well as establishing an achievable exponent, we derive an {\em ensemble tightness} result implying that one cannot do better with such a construction. 
  Proving this is non-trivial compared to previous ensemble tightness results (e.g., for random coding exponents \cite{SomekhBaruchMerhav2011,ScarlettMartinezGuilleniFabregas2012AllertonSU} and achievable rates \cite{MerhavKaplanLapidothShamai94}).  Among other things, the distribution of the drawn codeword depends on its index in the recursive construction, and on all of the previous codewords, so one cannot use a symmetry argument to focus on a single message.
    
    \item We show that when the distance function is optimized, the generalized RGV construction achieves the exponent of Csisz\'{a}r and K{\"o}rner \cite{CsiszarKorner81graph}, which is at least as high as both the random-coding and expurgated exponents.  While the analysis of \cite{CsiszarKorner81graph} establishes the existence of codes attaining the exponent using a decomposition lemma, our scheme provides a specific randomized construction that spreads the codewords according to a generic distance function, and whose ensemble average directly achieves the exponent.  
    
    \item In the case of an additive distance measure (e.g., Hamming or Bhattacharyya distance) and decoding metric (e.g., maximum-likelihood), we give an equivalent dual expression for our error exponent, as well as providing a direct derivation of the dual form using cost-constrained random coding \cite{ScarlettMartinezGuilleniFabregas_mismatch_2014_IT,ScarlettPengMerhavMartinezGuilleniFabregas_mismatch_2014_IT}.  
    This alternative derivation allows us to extend the achievability part to memoryless channels with infinite or continuous alphabets. 
    
    \item We prove that the distance function that measures closeness according to the joint empirical mutual information (equivalent to the equivocation distance \cite{BlahutComposition}) maximizes the exponent of our construction, at least among symmetric distance functions depending only on the joint type. This optimality is universal, in the sense that it holds for every channel and every type-dependent decoding metric. 
    In addition, we provide an alternative non-universal distance function yielding the same error exponent, and we show that an additive Chernoff-based distance measure (which reduces to the Bhattacharrya distance in the case of maximum-likelihood decoding) recovers both the random coding and expurgated exponents.
\end{itemize}

\subsection{Notation}\label{sc: definitions}
 
 The set of probability mass functions on a finite alphabet $\calX$ is denoted by $\mathcal{P}(\calX)$.  We use standard notations for entropy, mutual information, and so on (e.g., $I(X;Y)$, $H(X|Y)$), sometimes using a subscript to indicate the underlying distribution (e.g., $I_V(X;Y)$ for some joint distribution $V_{XY}$).  These are all taken to be in units of nats, and the function $\log$ has the natural base. We denote sequences (vectors) in boldfaced font, e.g., $\bx$. For $i<j$, we let $\bx_i^j$ denote $(\bx_i,\ldots,\bx_j)$, and similarly, $\bX_i^j=(\bX_i,\ldots,\bX_j)$. 
  
 We make frequent use of types \cite[Ch.~2]{CsiszarKorner81}.  The type (i.e., empirical distribution) of a sequence $\bx$ is denoted by $\hat{P}_{\bx}$, and similarly for joint types $\hat{P}_{\bx\by}$ and conditional types $\hat{P}_{\by|\bx}$.  The set of all types for a given sequence length $n$ is denoted by  $\mathcal{P}_n(\calX)$.  The type class $\calT(P)$ is the set of all sequences with type $P$, and the conditional type class $\calT(P_{\widetilde{X}|X})$ is the set of all $\widetilde{X}$-sequences inducing a given conditional type $P_{\widetilde{X}|X}$ for an arbitrary fixed $X$-sequence (whose type will be clear from the context).
 
 For two positive sequences $f_{n}$ and $g_{n}$, we write $f_{n}\doteq g_{n}$
 if $\lim_{n\to\infty}\frac{1}{n}\log\frac{f_{n}}{g_{n}}=0$, $f_{n}\,\dot{\le}\,g_{n}$
 if $\limsup_{n\to\infty}\frac{1}{n}\log\frac{f_{n}}{g_{n}}\le0$, and similarly for $\dot{\ge}$.

\subsection{Structure of the Paper}

In Section \ref{sc: Statement}, we formally introduce the channel coding setup and introduce additional notation.  In Section \ref{sc: Coding Scheme}, we describe the recursive random codebook construction and establish its main properties.  Section \ref{sc: Main Theorem} gives the main result and its proof, and Section \ref{sc:dual} gives the equivalent dual expression and its direct derivation. Section \ref{sc: Optimal Choices for d} studies the optimality of some specific distance functions.  
 \section{Problem Setup}\label{sc: Statement}

We consider the problem of reliable transmission over a DMC described by a conditional probability mass function $W(y|x)$, with input $x\in\calX$ and output $y\in\calY$ for finite alphabets $\calX$ and $\calY$. We define 
\begin{flalign}
W^n(\by|\bx) = \prod_{k=1}^n W(y_k|x_k)
\end{flalign}
for input/output sequences $\bx = (x_1,\dotsc,x_n)\in\calX^n,\by= (y_1,\dotsc,y_n)\in\calY^n$. We use the notation $\bX,\bY$ to denote the corresponding random variables. Infinite and continuous alphabets are addressed in Section \ref{sc:dual}.

An encoder maps a message $m\in \{1,\dotsc,M_n\}$ to a channel input sequence $\bx_m\in\calX$, 
where the number of messages is denoted by $M_n$.
The message, represented by the random variable $S$, is assumed to take values on $\{1,\dotsc,M_n\}$ equiprobably. This mapping induces an $(n,M_n)$-codebook $\calC_n=\{\bx_1,\dotsc,\bx_{M_n}\}$ with rate $R_n=\frac{1}{n}\log M_n$.

The decoder has access to the codebook and, upon observing the channel output $\by$, produces an estimate of the transmitted message $\hat m \in \{1,\dotsc,M_n\}$. 
We consider the family of maximum metric decoders for which the transmitted message is estimated as
\beq
\hat m = \argmax_{\bx_i\in\calC_n} q^{}(\bx_i,\by)
\eeq
where $q^{}(\bx,\by) : \calX^n\times\calY^n\to \reals$ is a generic decoding metric. Whenever two or more candidate codewords have the same decoding metric, an error will be assumed. Whenever $q^{}(\bx,\by)$ is an increasing function of the channel transition law $W^n(\by|\bx)$ we recover the maximum-likelihood (ML) decoder. Otherwise, the decoder is said to be mismatched \cite{MerhavKaplanLapidothShamai94,CsiszarNarayan95}. Throughout the paper, we assume that the decoding metric $q^{}(\bx,\by)$ only depends on the joint empirical distribution (or type) of $\bx,\by$, i.e., $\hat P_{\bx,\by}$. In this case, we write the decoder as
\beq
\hat m = \argmax_{\bx\in\calC_n} \,q(\hat{P}_{\bx,\by}),
\eeq
where we assume that the type-dependent metric $q:\mathcal{P}(\calX\times\calY)\to\reals$ is continuous (and therefore bounded) on the probability simplex.\footnote{Similarly to \cite{CsiszarKorner81graph}, our analysis easily extends to the ML decoding metric, for which $q(\bx,\by)$ may equal $-\infty$ when $W^n(\by|\bx)=0$.}  An important class of such metrics is the class of {\em additive metrics}, taking the form
\begin{equation}
    q(\hat{P}_{\bx,\by}) = \frac{1}{n} \sum_{i=1}^n q(x_i,y_i) = \EE_{\hat{P}_{\bx,\by}}[ q(X,Y) ], \label{eq:q_additive}
\end{equation}
where $q(x,y)$ is a {\em single-letter} metric (abusing notation slightly), and the average is with respect to the joint empirical distribution.  A notable example of a non-additive type-dependent metric is the empirical mutual information, $q(\hat{P}_{\bx,\by}) = I_{\hat{P}_{\bx,\by}}(X;Y)$.

Denoting the random variable corresponding to the decoded message by $\hat S$, we define the probability of error as
$
P_e = \Pr \bigl(\hat S \neq S\bigr).
$
A rate-exponent pair $(R,E)$ is said to be {\em achievable}
    for channel $W$ if, for all $\epsilon>0$, there exists a sequence of $(n,\e^{n(R-\epsilon)})$-codebooks such that 
    \begin{flalign}
        \liminf_{n\to\infty} \, -\,\frac{1}{n}\log \Pr\bigl(\hat{S}\neq S\bigr)\geq E-\epsilon. 
    \end{flalign}
    Equivalently, we say that $E$ is an achievable error exponent at rate $R$ if $(R,E)$ is an achievable rate-exponent pair.

\section{Random Codebook and Properties} \label{sc: Coding Scheme} 

In this section, we introduce our recursive random coding scheme, and state its main properties used for deriving the associated error exponent. 

Codes that attain the Gilbert-Varshamov bound on the Hamming space \cite{gilbert1952comparison,varshamov1957estimate} ensure that all codewords are at least at a certain target Hamming distance $\Delta$ from each other.  The generalized RGV construction is a randomized constant-composition counterpart of such codes for arbitrary DMCs and more general distance functions. 

\begin{definition} \label{def:distance}
    Let $\Omega$ be the set of bounded, continuous, symmetric, and type-dependent functions $d(\cdot,\cdot):\calX^n\times \calX^n \to \reals$, i.e., bounded functions that satisfy $d(\bx,\bx')=d(\bx',\bx)$ for all $\bx,\bx'\in \calX^n$, that depend on $(\bx,\bx')$ only through the joint empirical distribution $\hat{P}_{\bx\bx'}$, and that are continuous on the probability simplex.
\end{definition}
We use the notation $d(\bx,\bx')$ and $d(\hat{P}_{\bx\bx'})$ interchangeably for convenience, similarly to $q(\bx,\by)$ and $q(\hat{P}_{\bx\by})$.  We refer to $d\in\Omega$ as a {\em distance function}, though it need not be a distance in the topological sense (e.g., it may be negative).

Some examples of distance functions in $\Omega$ are as follows:
\begin{itemize}
    \item We say that the distance function is {\em additive} if it can be written as
    \begin{equation}
        d(\bx,\bx') = \frac{1}{n} \sum_{k=1}^n d(x_k,x'_k) \label{eq:additive_d}
    \end{equation}
    for some single-letter function $d(x,x')$ (abusing notation slightly).  Any such distance function is in $\Omega$, as long as $d(x,x')$ is symmetric.  Notable examples include the Hamming distance
    \begin{align}
        d_{\text H}(x,x') &= \indicator\{x'\neq x\},
    \end{align}
    and the Bhattacharyya distance
    \begin{align}
        d_{\text B}(x,x') &= -\log  \sum_{y\in\calY}\sqrt{W(y|x)W(y|x')}. \label{eq:Bhat}
    \end{align}
    Note that the latter choice depends on the channel, and to satisfy the boundedness assumption we require that any two inputs have a common output that is produced with positive probability.
    \item We will later consider a distance equal to the negative mutual information, $d(P_{X\widetilde{X}}) = -I_P(X;\widetilde{X})$, which will turn out to be universally optimal subject to the constraints of our construction. For constant-composition codes, it is equivalent to the {\em equivocation distance} $d(P_{X\widetilde{X}}) = H_P(\widetilde{X}|X)$, which was considered in a different but related context by Blahut \cite{BlahutComposition}. 
\end{itemize}

In the following, we describe how to construct a code $\calC_n$ with $M_n$ codewords of length $n$, such that any two distinct codewords $\bx,\bx'\in \calC_n$ satisfy $d(\bx,\bx') > \Delta$ for a given function $d(\cdot,\cdot)\in\Omega$ and threshold $\Delta\in\reals$. This guarantees that the minimum distance of the codebook exceeds $\Delta$.  The construction depends on an input distribution $P \in \calP(\calX)$, and throughout the paper, we let $P_n \in \calP_n(\calX)$ denote an arbitrary type with the same as support as $P$ satisfying $\max_{x\in\calX} |P_n(x) - P(x)| \le \frac{1}{n}$.

Along with $P\in \mathcal{P}(\calX)$, fixing $n,M_n$, a distance function $d(\cdot,\cdot)\in\Omega$, and constants $\delta>0, \Delta\in\mathbb{R}$, the construction is described by the following steps:
\begin{enumerate}
\item The first codeword, $\bx_1$, is drawn uniformly from $\calT(P_n)$;
\item The second codeword $\bx_2$ is drawn uniformly from 
\begin{align}
        \calT(P_n,\bx_1)&\triangleq\left\{\bar{\bx}\in \calT(P_n) \,:\, d(\bar{\bx},\bx_1)> \Delta\right\}\\
         &=\calT(P_n)\backslash \left\{\bar{\bx}\in \calT(P_n) \,:\, d(\bar{\bx},\bx_1)\leq \Delta\right\},
     \end{align}
i.e., the set of sequences with composition $P_n$ whose distance to $\bx_1$ exceeds $\Delta$;
\item Continuing recursively, the $i$-th codeword $\bx_i$ is drawn uniformly from
 \begin{align}
         &\calT(P_n,\bx_1^{i-1}) \nonumber \\
         &\triangleq\left\{\bar{\bx}\in \calT(P_n)  \,:\,  d(\bar{\bx},\bx_j)> \Delta, \, j=1\dotsc, i-1\right\}\\
         &=\calT(P_n,\bx_1^{i-2})
         \backslash \left\{\bar{\bx}\in \calT(P_n,\bx_1^{i-2}) \,:\, d(\bar{\bx},\bx_{i-1})\leq \Delta\right\}.
     \end{align}
\end{enumerate}
Throughout the paper, it will be useful to generalize the notation $\calT(P_n, \bx_1^{i-1})$ as follows. For any subset  $\calD\subseteq \calT(P_n)$, we define
\begin{flalign}\label{eq: cals calA dfn}
\calT(P_n,\calD)\triangleq \{\bx\in\calT(P_n):\; d(\bx, \bx')>\Delta,\, \forall \bx'\in\calD\}.
\end{flalign}

In Lemma \ref{lem:Ti_bounds} below, we will show that in order to ensure that the above procedure generates the desired number of codewords $M_n=\e^{nR_n}$ (i.e., the sets $\calT(P_n,\bx_1^{i-1})$ are non-empty for all $i=1,\dotsc,M_n$), it suffices to choose $\Delta$ and $\delta$ such that 
\beq
\e^{n(R_n+\delta)} \vol_{\bx}(\Delta) \leq |\calT(P_n)|
\label{eq:vol_condition}
\eeq
where
$
\vol_{\bx}(\Delta) = |\{\bar{\bx}\in \calT(P_n) \,:\, d(\bar{\bx},\bx)\leq \Delta\}|
$
is the ``volume'' of a ``ball'' of radius $\Delta$ according to the ``distance'' $d(\cdot,\cdot)$, centered at some $\bx\in \calT(P_n)$.  Since $d\in\Omega$ is symmetric and type-dependent, $\vol_{\bx}(\Delta)$ does not depend on the specific choice of $\bx\in \calT(P_n)$. 
It will be convenient to rewrite \eqref{eq:vol_condition} as
\begin{flalign}
    \sum_{\bar{\bx}\in \calT(P_n) \,:\, d(\bar{\bx}, \bx) \leq\Delta }\frac{1}{|\calT(P_n)|} \leq  \e^{-n(R_n+\delta)}.
    \label{eq: Delta dfn stricter condition}
\end{flalign}

\subsection{Codebook Properties}

Here we provide several lemmas characterizing the key properties of the generalized RGV construction.  
We begin with the fact that the construction is well-defined, in the sense that  the procedure described above always produces the desired number of codewords $M_n$, i.e., the set $\calT(P_n, \bx_1^{i-1})$ given the previous codewords is always non-empty.

\begin{lemma} \label{lem:Ti_bounds}
    The generalized RGV codebook construction with condition \eqref{eq: Delta dfn stricter condition} is such that for all $i\in\{1,\dotsc,M_n\}$, all $\bx_1^{i-1}$ occurring with non-zero probability, and any $\delta>0$, we have
        \begin{flalign}\label{eq: T card bounds}
   \hspace{-3mm}         (1-\e^{-n\delta})|\calT(P_n)|\leq |\calT(P_n, \bx_1^{i-1})|\leq |\calT(P_n)|. 
        \end{flalign}
\end{lemma}

\begin{proof}
    The upper bound is trivial, since
    \begin{align}
  \calT(P_n, \bx_1^{M_n-1})\subseteq \cdots \subseteq  \calT(P_n, \bx_1^{i-1})&\subseteq \calT(P_n, \bx_1^{i-2}) \nonumber\\
  & \subseteq \cdots\subseteq \calT(P_n).
  \label{eq:subsets}
    \end{align}
 For the lower bound, we make use of \eqref{eq:vol_condition}--\eqref{eq: Delta dfn stricter condition}. After $M_n =\e^{nR_n}$ iterations of the above procedure, we have removed no more than $\e^{n R_n} \vol_{\bx}(\Delta)\leq  |\calT(P_n)|\e^{-n\delta}$ sequences from $\calT(P_n)$. This implies that after iteration $M_n =\e^{nR_n}$, 
\begin{align}
    |\calT(P_n, \bx_1^{M_n-1})| &\geq |\calT(P_n)| - \e^{n R_n} \vol_{\bx}(\Delta)\\
    &\geq  |\calT(P_n)|(1-\e^{-n\delta}).\label{eq:lb}
\end{align}
The lower bound in (\ref{eq: T card bounds}) for $i\in\{1,\dotsc,M_n\}$ follows from \eqref{eq:lb} and \eqref{eq:subsets}.

\end{proof}

Henceforth, whenever we refer to the generalized RGV construction, this implicitly includes the condition \eqref{eq: Delta dfn stricter condition} (or equivalently \eqref{eq:vol_condition}).

The following lemmas provide upper and lower bounds on the marginal distributions of small numbers of codewords (up to three) in the RGV construction.  We make use of the following exponentially vanishing quantity:
\begin{flalign}\label{eq: delta_n dfn }
    \delta_n\triangleq  \frac{\e^{-n\delta}}{1-\e^{-n\delta}}.
\end{flalign}
We begin with the joint distribution between two codewords, as this plays the most important role in our analysis.  Here and subsequently, $\Pr(\bx_k,\bx_m)$ is a shorthand for $\Pr(\bX_k = \bx_k, \bX_m = \bx_m)$, and similarly for other expressions such as $\Pr(\bx_i,\bx_j,\bx_k)$.

\begin{lemma}\label{lm: pairwise distribution of codewords}
Under the generalized RGV construction, for any $k \in \{1,\dotsc,M_n - 1\}$, $m>k$ and $\bx_k,\bx_m \in \calT(P_n)$, if $d(\bx_k,\bx_m)>\Delta$ then we have
\begin{flalign}
   \frac{(1-4\delta_n^2)}{|\calT(P_n)|^2} \e^{-2\delta_n}  \leq \Pr(\bx_k,\bx_m)\leq \frac{1}{(1-\e^{-n\delta})^2|\calT(P_n)|^2},\label{eq: pairwise distribution of codewords k m}
\end{flalign}
while $\Pr(\bx_k,\bx_m) = 0$ whenever $d(\bx_k,\bx_m) \le \Delta$.
\end{lemma}
\begin{proof}
    See Appendix \ref{app:pairwise lemma}.
\end{proof}

Note that here $k$ and $m$ are arbitrary indices, and $m$ need not correspond to the transmitted message.  In some cases, we will apply the lemma with $m$ being the transmitted message.

For the joint distribution between three codewords, we will only require an upper bound, and it will only be used for the ensemble tightness part.

\begin{lemma}\label{lm: triplets-wise distribution of codewords}
Under the generalized RGV construction, for any $i,j,k \in \{1,\dotsc,M_n\}$, such that $i<j<k$ and $\bx_i,\bx_j ,\bx_k\in \calT(P_n)$, if $\min\{d(\bx_i,\bx_j),d(\bx_i,\bx_k),d(\bx_j,\bx_k)\}>\Delta$ then
\begin{flalign}
  \Pr(\bx_i,\bx_j,\bx_k)\leq \frac{1}{(1-e^{-n\delta})^3|\calT(P_n)|^3},\label{eq: tiplets-wise distribution of codewords imj}
\end{flalign}
while $\min\{d(\bx_i,\bx_j),d(\bx_i,\bx_k),d(\bx_j,\bx_k)\}\le \Delta$ whenever $\Pr(\bx_i,\bx_j,\bx_k) = 0$.
\end{lemma}
\begin{proof}
    See Appendix \ref{app:triplets distribution lemma}.
\end{proof}

Finally, by a basic symmetry argument, the marginal distribution of any given codeword $\bX_m$ (without any conditioning) is uniform over $\calT(P_n)$, as stated in the following.

\begin{lemma}\label{lm: marginal distribution Xm lemma}
    For any message index $m$, the marginal distribution of codeword $\bX_m$ is $\Pr(\bx_m) = \frac{1}{|\calT(P_n)|}$ for $\bx_m \in \calT(P_n)$, and zero elsewhere.
\end{lemma}
\begin{proof}
    See Appendix \ref{ap: Proof of Uniformity}.
\end{proof}

\section{Main Result} \label{sc: Main Theorem}

Using graph decomposition techniques, Csisz\'{a}r and K\"{o}rner \cite{CsiszarKorner81graph} studied the error exponents of constant-composition codes under a decoder that uses a type-dependent decoding metric $q(\hat{P}_{\bx,\by})$, and derived the following achievable exponent for an arbitrary input distribution $P$:
    \begin{flalign}\label{eq: E_ex dfn C and K}
        E_{q}(R,P,W)&= \min_{V\in\calT_I }D(V_{Y|X}\|W|P)+\big|I(\widetilde{X};Y,X)-R\big|_+,
    \end{flalign}
    where
    \begin{flalign}\label{eq: cal T I definition}
        &\calT_I\triangleq \Big\{ V_{X\widetilde{X}Y} \in\calP(\calX\times\calX\times\calY)\,:\, \nonumber  \\
        & \quad V_X=V_{\widetilde{X}}=P,q(V_{\widetilde{X}Y}) \geq q(V_{XY}), I(X;\widetilde{X})\leq R \Big\}.
    \end{flalign}
    This exponent was shown to be at least as high as the maximum of the expurgated exponent and the random coding exponent.

The following theorem presents an exact single-letter expression for the error exponent of the recursive RGV codebook construction described in the previous section.  We show in Section \ref{sc: Optimal Choices for d} that it reduces to the exponent of \cite{CsiszarKorner81graph}, $E_{q}(R,P,W)$, when the distance function $d(\cdot,\cdot)$ is optimized. 

Letting
    \begin{flalign}\label{eq: E_ex dfn}
        &E_{\mathrm{RGV}}(R,P,W,q,d, \Delta) \nonumber \\
        &= \min_{V_{X\widetilde{X}Y}\in\calT_{d,q,P}(\Delta) }D(V_{Y|X}\|W|P)+\bigl|I(\widetilde{X};Y,X)-R\bigr|_+,
    \end{flalign}
    where
    \begin{flalign}\label{eq: cal T alpha definition}
        &\calT_{d,q,P}(\Delta) \triangleq\Bigl\{
        V_{X\widetilde{X}Y}\in\mathcal{P}(\calX\times\calX\times\calY) \,:\, \nonumber \\ 
        &V_X=V_{\widetilde{X}}=P,\, q(V_{\widetilde{X}Y})\geq q(V_{XY}), \,d(V_{X\widetilde{X}})\geq \Delta \Bigr\},
    \end{flalign}
we have the following.

\begin{theorem}\label{th: main theorem sphere packing}
    For all $P\in\mathcal{P}(\calX)$, $\delta>0$, $\Delta\in\reals$, $d\in\Omega$, and $R>0$ satisfying 
        \begin{flalign}\label{eq: frist eq}
        R\leq \min_{P_{X\widetilde{X}} \,:\, d(P_{X\widetilde{X}}) \leq\Delta,\, P_X=P_{\widetilde{X}}=P} I(X;\widetilde{X})-2\delta,
    \end{flalign}
     the ensemble average error probability $\bar{P}_\e^{(n)}$ of the generalized RGV construction with parameters $(n,R,P,d, \Delta,\delta)$ and the bounded continuous type-dependent decoding metric $q(\cdot)$ over the DMC $W$ satisfies
    \begin{flalign}
       \bar{P}_\e^{(n)}\dotleq  \e^{-nE_{\mathrm{RGV}}(R,P,W,q,d,\Delta)}.
    \end{flalign}
    In addition, if $q$ is an additive decoding metric, then
    \begin{equation}\label{eq: ensemble tightness Theorem}
        \bar{P}_\e^{(n)} \dotgeq \e^{-nE_{\mathrm{RGV}}(R,P,W,q,d,\Delta+\epsilon)} 
    \end{equation}
    for arbitrarily small $\epsilon>0$. 
\end{theorem}

The achievability proof (i.e., upper bound on the error probability) is given in Section \ref{sc: upper bound}, and the ensemble tightness proof (i.e., lower bound on the error probability) for additive metrics is given in Section \ref{sc: lower bound}.

As will be shown in Section \ref{sc: Optimal Choices for d}, under the rate constraint \eqref{eq: frist eq}, if the distance function is chosen appropriately, the generalized RGV construction achieves the exponent $E_{q}(R,P,W)$ in \eqref{eq: E_ex dfn C and K}, which in turn shows the achievability of capacity for ML decoding or the LM rate in the mismatched case \cite{CsiszarKorner81graph,Hui83}. Moreover, for a distance function $d$ that uniquely attains its minimum value when $X=X'$, varying $\Delta$ from its minimum to maximum value yields all possible values of rates in $(0,H(P))$, which covers the entire range of possible rates with constant composition codes.

Theorem \ref{th: main theorem sphere packing} implies that the exact exponent of the coding scheme equals $E_{\mathrm{RGV}}(R,P,W,q,d,\Delta)$ whenever $\Delta$ is a continuity point.  We note that while the additivity of $q(\cdot)$ is required for the derivation of the lower bound on $ \bar{P}_\e^{(n)}$, the upper bound holds also for any continuous $q(\cdot)$ that need not be additive. 
        For non-additive $q$, in the 
        assertion of the lower bound \eqref{eq: ensemble tightness Theorem} for non additive $q$, we would have to replace $E_{\mathrm{RGV}}(R,P,W,q,d,\Delta+\epsilon) $ by
    \begin{flalign}\label{eq: E_ex dfn jdhgvfkh}
        \min_{V_{X\widetilde{X}Y}\in\mathcal{T}_{d,q,P,\epsilon}(\Delta+\epsilon) }D(V_{Y|X}\|W|P)+\bigl|I(\widetilde{X};Y,X)-R\bigr|_+,
    \end{flalign}
    where
     \begin{flalign}\label{eq: cal T I definition tag}
        &\mathcal{T}_{d,q,P,\epsilon}(\Delta+\epsilon) \triangleq \Big\{ V_{X\widetilde{X}Y} \in\mathcal{P}(\mathcal{X}\times\mathcal{X}\times\mathcal{Y})\,:\, \nonumber \\
        &V_X=V_{\widetilde{X}}=P, q(V_{\widetilde{X}Y}) \geq q(V_{XY})+\epsilon,\; d(V_{X\widetilde{X}}) \geq\Delta+\epsilon \Big\},
    \end{flalign}
    that is, we would have the extra $\epsilon$ in $q(V_{\widetilde{X}Y}) \geq q(V_{XY})+\epsilon$. 
    While this yields the desired tightness result whenever the optimization 
is "continuous" with respect to the metric constraint, it is unclear in 
what generality such continuity holds.  As a simple example, $\mathcal{T}_{d,q,P,\epsilon}(\Delta+\epsilon)$ is always empty under the erasures-only metric $q(x,y) = \indicator\{ W(y|x) > 0\}$, meaning that \eqref{eq: E_ex dfn jdhgvfkh} does not provide a tightness result in this case.

By a simple symmetrization argument, we can show that $E_{\mathrm{RGV}}(R,P,W,q,d,\Delta)$ is an achievable error exponent even when $d$ is not symmetric.  This is stated in the following.

\begin{corollary} \label{cor:asymm}
    Under the setup of Theorem \ref{th: main theorem sphere packing} with a non-symmetric continuous type-dependent bounded distance function $d$, if the pair $(R,\Delta)$ satisfies \eqref{eq: frist eq}, then the error exponent 
    
\noindent $E_{\mathrm{RGV}}(R,P,W,q,d,\Delta)$ is achievable at rate $R$. 
\end{corollary}
\begin{proof}
    We apply Theorem \ref{th: main theorem sphere packing} with the symmetric distance
    \begin{equation}
        d'(\bx,\bx^{\prime})=\min\big\{ d(\bx,\bx^{\prime}),d(\bx^{\prime},\bx)\big\}. \label{eq:symmetrized}
    \end{equation}
    Notice that this choice enforces $d(\bx,\bx^{\prime}) > \Delta$ for all pairs $(\bx_{i},\bx_{j})$ in the codebook, regardless of whether $i<j$ or $i>j$.
    
    The exponent in \eqref{eq: E_ex dfn} with symmetric distance $d'$ simplifies as follows:
    \begin{align}
    & \min_{\substack{V:V_{X}=V_{\widetilde{X}}=P,\\
    q(V_{\widetilde{X}Y})\ge q(V_{XY}),\\\min\{d(P_{X\widetilde{X}}),d(P_{\widetilde{X}X})\}\ge\Delta}}D(V_{Y|X}\|W|P)+[I(\widetilde{X};X,Y)-R]_{+} \nonumber \\
    & \quad\ge\min_{\substack{V:V_{X}=V_{\widetilde{X}}=P,\\
    q(V_{\widetilde{X}Y})\ge q(V_{XY}),\\ d(P_{X\widetilde{X}})\ge\Delta}}D(V_{Y|X}\|W|P)+[I(\widetilde{X};X,Y)-R]_{+},
    \end{align}
    since on the right-hand side we are minimizing over a larger set. 
    Moreover, the minimization in the rate condition \eqref{eq: frist eq} with distance $d'$ simplifies as follows:
    \begin{align}
    & \min_{P_{X\widetilde{X}}\,:\,P_{X}=P_{\widetilde{X}}=P,\min\{d(P_{X\widetilde{X}}),d(P_{\widetilde{X}X})\}\le\Delta}I(X;\widetilde{X}) \nonumber \\
    & =\min\bigg\{
    \min_{P_{X\widetilde{X}}\,:\,P_{X}=P_{\widetilde{X}}=P,d(P_{X\widetilde{X}})\le\Delta}I(X;\widetilde{X}),\nonumber\\ 
    &\qquad \qquad\qquad\min_{P_{X\widetilde{X}}\,:\,P_{X}=P_{\widetilde{X}}=P,d(P_{\widetilde{X}X})\le\Delta}I(X;\widetilde{X})    \bigg\}\\
    & \qquad=\min_{P_{X\widetilde{X}}\,:\,P_{X}=P_{\widetilde{X}}=P,d(P_{X\widetilde{X}})\le\Delta}I(X;\widetilde{X}),
    \end{align}
    where the second line follows since $\min_{z\in A\cup B}f(z)=\min\big\{\min_{z\in A}f(z),\min_{z\in B}f(z)\big\}$,
    and the last line follows from the symmetry of mutual information.
\end{proof}

We briefly discuss the proof of Theorem \ref{th: main theorem sphere packing}.  While the theorem states the error exponent, the central part of the analysis is in arriving at the following asymptotic expression for the ensemble average probability of error: \begin{flalign}
    &\hspace{-2mm}\bar{P}_\e^{(n)}\doteq   \sum_{\bx \in \calT(P_n),\by}\frac{1}{|\calT(P_n)|}W^n(\by|\bx)\nonumber\\
    &\times\min\Biggl\{1,(M_n-1) \sum_{\substack{\bx' \in \calT(P_n) \,:\, q^{}(\bx',\by)\geq q^{}(\bx,\by)\\ d(\bx',\bx)\geq \Delta }} \frac{1}{|\calT(P_n)|} \Biggr\}, \label{eq:rcu-like}
\end{flalign}
which holds for every type-dependent decoding metric $q$ (not necessarily additive or continuous). 
This can be interpreted as a stronger (albeit asymptotic) analog of the {\em random coding union} bound \cite{Polyanskiy10} that achieves not only the random coding exponent, but also the low-rate improvements of the expurgated exponent. 

It is also worth discussing the connection of Theorem \ref{th: main theorem sphere packing} with the analysis of \cite{CsiszarKorner81graph} based on graph decomposition techniques.  A key result shown therein is the existence of a rate-$R$ constant composition codebook $\calC_n$ such that each $\bx\in\calC_n$ satisfies
        \begin{equation}
        |\calT_{\bar V}(\bx) \cap \calC_n|\leq \exp\{n(R-I(P,\bar V))\} \label{eq:Vbar}
        \end{equation}
        for all conditional types $\bar V$ representing a ``channel'' from $\mathcal{X} \to \mathcal{X}$, where $I(P, \bar V) = I_{P \times \bar V}(X;X')$.  In the derivation of $E_q$ ({\em cf.}, \eqref{eq: E_ex dfn C and K}), \eqref{eq:Vbar} is used to establish the empirical mutual information bound $I_{\hat{P}_{\bx,\bx'}}(X;X') \le R$ for any two codewords $\bx,\bx' \in \calC_n$.  
        It is also used to upper bound the number of output sequences $\by$ that can give rise to a given joint type $\tilde{P}_{XX'Y}$, with \eqref{eq:Vbar} characterizing the $\tilde{P}_{XX}$ marginal and standard techniques characterizing $\tilde{P}_{Y|XX'}$.

        Although it was not shown in \cite{CsiszarKorner81graph}, \eqref{eq:Vbar} can be used to establish the achievability part of Theorem \ref{th: main theorem sphere packing} for general distance functions.  To see this, let $I_{\min}$ be the smallest empirical mutual information among codeword pairs $(\bx;\bx')$ with $d(\bx;\bx') \le \Delta$, as stated in Theorem  \ref{th: main theorem sphere packing}.  If $R < I_{\min}$, then the left-hand side of \eqref{eq:Vbar} is zero, meaning all codeword pairs satisfy $d(\bx;\bx') > \Delta$.  Upon noticing this fact, the rest of the proof of \cite[Theorem 1]{CsiszarKorner81graph} remains essentially unchanged and yields the RGV exponent.
        
        Compared to \cite{CsiszarKorner81graph} and other related works, the main advantages of our approach are as follows: (i) We provide an explicit recursive random coding construction rather than only proving existence; (ii) We establish, to our knowledge, the first ensemble tightness result for any construction achieving the expurgated exponent; (ii) We provide a direct extension to channels with continuous alphabets, whereas \cite{CsiszarKorner81graph} relies heavily on combinatorial arguments and types.

\subsection{Proof of Achievability (Upper Bound on $\bar{P}_\e^{(n)}$)}\label{sc: upper bound}

The proof is given in three steps.

\vspace{0.1cm}
\noindent {\em \underline{Step 1: Characterizing the permitted rates}}
\vspace{0.1cm}

\noindent 
For convenience, we define
\begin{flalign}\label{eq: frist eq B}
    R'\triangleq \min_{P_{X\widetilde{X}} \in \mathcal{P}(\calX^2) \,:\, d(P_{X\widetilde{X}}) \leq\Delta,\, P_X=P_{\widetilde{X}}=P} I(X;\widetilde{X})-2\delta.
\end{flalign}
Recalling that $\calT(P_{\widetilde{X}|X})$ stands for a conditional type class \cite[Ch.~2]{CsiszarKorner81} corresponding to $\bx \in \calT(P_n)$, and  letting $\mathcal{P}_n(\calX|\bx)$ be the set of all conditional types, 
we have for $n$ sufficiently large that\begin{flalign}
   & \sum_{\bar{\bx}\in \calT(P_n) \,:\, d(\bar{\bx},\bx)\leq\Delta} \frac{1}{|\calT(P_n)|}\notag\\
    &~~~~\leq (n+1)^{|\calX|^2}\max_{\substack{P_{\widetilde{X}|X}\in \mathcal{P}_n(\calX|\bx) \,:\, P_{\widetilde{X}}=P_X=P_n\\ d(P_{X\widetilde{X}})\leq \Delta}}\frac{|\calT(P_{\widetilde{X}|X})|}{|\calT(P_n)|} \label{eq:deltaub1}\\
   &~~~~ \leq   \exp\bigg(-n\bigg(\min_{\substack{P_{X\widetilde{X}}\in \mathcal{P}_n(\calX^2) \,:\, d(P_{X\widetilde{X}}) \leq\Delta \\P_X=P_{\widetilde{X}}=P}} I(X;\widetilde{X})-\delta \bigg)\bigg) \label{eq:deltaub2} \\
    &~~~~\leq  \e^{-n(R'+\delta)}\label{eq:deltaub3},
\end{flalign}
where \eqref{eq:deltaub1} follows since the number of conditional types is upper bounded by $(n+1)^{|\calX|^2}$, \eqref{eq:deltaub2} holds for $n$ sufficiently large because $|\calT(P_{\widetilde{X}|X})| \doteq \e^{nH_P(\widetilde{X}|X)}$ and $|\calT(P_n)|\doteq \e^{nH(P)}$ \cite[Ch.~2]{CsiszarKorner81}, and \eqref{eq:deltaub3} follows from \eqref{eq: frist eq B} and the fact that $\mathcal{P}_n(\calX^2)\subseteq \mathcal{P}(\calX^2)$. Hence, if the rate of the generalized RGV construction satisfies $R_n\leq R'$, we have 
\beq\label{eq: how ccc} 
\sum_{\bar{\bx}\in \calT(P_n) \,:\, d(\bar{\bx},\bx)\leq\Delta} \frac{1}{|\calT(P_n)|}\leq \e^{-n(R_n+\delta)},
\eeq 
which is  precisely the condition assumed in \eqref{eq: Delta dfn stricter condition}.

We henceforth assume that the number of codewords of the generalized RGV construction is such that $R_n\leq R'$, and calculate the resulting average probability of error.

\vspace{0.1cm}
\noindent {\em \underline{ Step 2: Conditional error probability}}
\vspace{0.1cm}

\noindent We define the $i$-th pairwise error event given $(\bX_m,\bY)=(\bx_m,\by)$, where $i\neq m$ as
\begin{flalign}
    \calE_i = \left\{ q^{}(\bX_i, \by) \geq q^{}(\bx_m,\by)\right\}, \label{eq:E_i}
\end{flalign}
meaning that the random codeword $\bX_i$ is favored over $\bx_m$ (or the two are favored equally). The ensemble average error probability is 
\beq
\bar P_e^{(n)} = \frac{1}{M_n}\sum_{m=1}^{M_n} \bar P_{e,m}^{(n)},\label{eq: avg err prob}
\eeq
where the probability of error assuming that the $m$-th codeword has been transmitted is
\begin{align}
\bar P_{e,m}^{(n)} &=\EE[\Pr(\text{error}\,|\,\bX_m,\bY) ],
\end{align}
and where
\begin{align}
\Pr(\text{error}\,|\,\bx_m,\by) &=  \Pr\Bigg(\bigcup_{\substack{i=1\\ i\neq m}}^{M_n} \calE_i
 \,\bigg|\, \bX_m = \bx_m,\bY=\by\Bigg) \label{eq:pem2}
\end{align}
is the probability of decoding error for the $m$-th codeword assuming that the realizations of the codeword and received sequences are $\bx_m$ and $\by$ (recall that ties are counted as errors).  We initially perform the analysis conditioned on the transmitted and received sequences being $\bx_m$ and $\by$, respectively (and implicitly on $m$ being transmitted), and later we duly average over these choices.

Now, since only sequences $\bx_i$ such that $d(\bx_i,\bx_m)> \Delta$ have positive probability conditioned on $\bX_m=\bx_m$, we have 
\begin{flalign}
    &\Pr(\calE_i | \bx_m,\by) \nonumber \\
    &= \sum_{\substack{
    \bx_i \,:\,  q^{}(\bx_i,\by)\geq q^{}(\bx_m,\by) \\d(\bx_i,\bx_m)> \Delta}} \Pr(\bx_i| \bx_m,\by)\\
    &= \sum_{\substack{
    \bx_i \,:\,  q^{}(\bx_i,\by)\geq q^{}(\bx_m,\by) \\d(\bx_i,\bx_m)> \Delta}} \Pr(\bx_i| \bx_m)
     \label{eq:case_1}\\
    &= \sum_{\substack{
    \bx_i \,:\,  q^{}(\bx_i,\by)\geq q^{}(\bx_m,\by) \\d(\bx_i,\bx_m)> \Delta}} \frac{ \Pr(\bx_i, \bx_m)}{\Pr( \bx_m)}
     \label{eq:case_1 rev}\\
    &\leq \frac{1}{(1-\e^{-n\delta})^2}\sum_{\substack{
    \bx_i \in \calT(P_n) \,:\,  q^{}(\bx_i,\by)\geq q^{}(\bx_m,\by) \\d(\bx_i,\bx_m)> \Delta}} \frac{1}{|\calT(P_n)|}    , \label{eq:case_5 rev}
\end{flalign}
where \eqref{eq:case_1} follows since $\bX_i - \bX_m-\bY$ forms a Markov chain, and \eqref{eq:case_5 rev} follows from Lemmas \ref{lm: pairwise distribution of codewords} and \ref{lm: marginal distribution Xm lemma}.

Applying the union bound to \eqref{eq:pem2} and substituting \eqref{eq:case_5 rev}, we obtain
\begin{align}
&\Pr(\text{error}\,|\,\bx_m,\by) \nonumber \\
&\leq \sum_{\substack{i \in \{1,\dotsc,M_n\},\\ i\neq m}}\Pr\big( \calE_i
 \,\big|\, \bX_m = \bx_m,\bY=\by\big) \\
 &\leq  \frac{1}{(1-\e^{-n\delta})^2}\sum_{\substack{i \in \{1,\dotsc,M_n\},\\ i\neq m}}
~ \sum_{
 \substack{\bx_i \in \calT(P_n)\,:\\q^{}(\bx_i,\by)\geq q^{}(\bx_m,\by)\\d(\bx_i,\bx_m)\geq \Delta }} \frac{1}{|\calT(P_n)|} \label{eq:pem4 rev}\\
&  =  (M_n-1)\frac{1}{(1-\e^{-n\delta})^2}
 \sum_{\substack{
    \bx'\in\calT(P_n) \,:\,\\  q^{}(\bx',\by)\geq q^{}(\bx_m,\by) \\d(\bx',\bx_m)> \Delta}} \frac{1}{|\calT(P_n)|},
 \label{eq:pem3 rev}
\end{align}
where \eqref{eq:pem3 rev} follows since summands in the summation over $\bx_i$ are equal for all $i$.  

 Applying the obvious inequality $\Pr(\text{error}\,|\,\bx_m,\by)\leq 1$, and slightly enlarging the set of summands by replacing $d(\bx',\bx)> \Delta$ by $d(\bx',\bx)\geq \Delta$, 
 it follows that
\begin{flalign}
\bar{P}_\e^{(n)}
&\dotleq  \sum_{\bx \in \calT(P_n),\by}\frac{1}{|\calT(P_n)|}W^n(\by|\bx)\notag\\&
\times 
 \min\Bigg\{1,(M_n-1) \sum_{\substack{\bx' \in \calT(P_n)\,:\\q^{}(\bx',\by)\geq q^{}(\bx,\by)\\d(\bx',\bx)\geq \Delta }} \frac{1}{|\calT(P_n)|} \Bigg\}, \label{eq:final_rcu}
\end{flalign}
where we have averaged over $(\bx_m,\by)$ and used Lemma \ref{lm: marginal distribution Xm lemma}.

\vspace{0.1cm}
\noindent {\em \underline{Step 3: Deducing the error exponent}}
\vspace{0.1cm}

\noindent Deducing the error exponent from \eqref{eq:final_rcu} amounts to a standard analysis based on the method of types, so we provide a rather brief treatment.

Similarly to \eqref{eq:deltaub1}, the inner sum in \eqref{eq:final_rcu} satisfies
\begin{flalign}
    \sum_{\substack{\bx' \in \calT(P_n):\\q^{}(\bx',\by)\geq q^{}(\bx,\by),\\ d(\bx',\bx)\geq \Delta }} \frac{1}{|\calT(P_n)|} \dotleq 
    \max_{\substack{\hat{P}_{\bx'|\bx\by}\in\calP_n(\calX|\bx\by) :\\q^{}(\bx',\by)\geq q^{}(\bx,\by)\\ d(\bx',\bx)\geq \Delta}} \frac{|\calT(\hat{P}_{\bx'|\bx\by})|}{|\calT(P_n)|}. \label{eq:inner_sum}
\end{flalign}
Applying the standard properties of types $|\calT(\hat{P}_{\bx'|\bx\by})| \doteq \e^{nH_{\hat{P}}(\widetilde{X}|Y,X)}$ and $|\calT(P_n)| \doteq \e^{nH(P_n)}$ \cite[Ch.~2]{CsiszarKorner81}, we can simplify the objective on the right-hand side of \eqref{eq:inner_sum} to $e^{-n I(\widetilde{X};X,Y)}$.  Moreover, we have
$
W^n(\by|\bx) = \e^{n(D(\hat{P}_{\by|\bx}\|W|P_n) +H( \hat{P}_{\by|\bx}) )},
$
which implies that $(\bX_m,\bY)$ has a given conditional type $V_{Y|X}$ with probability $\e^{-nD(V_{Y|X}\|W|P_n)}$ times a subexponential factor.  Using the following continuity lemma to replace $P_n$ by its limiting value $P$, we deduce the final single-letter exponent:
\begin{flalign}
\hspace{-1mm}\bar{P}_\e^{(n)}\dotleq   \e^{-n\min_{V\in\calT_{d,q,P}(\Delta) }D(V_{Y|X}\|W|P)+|I(\widetilde{X};Y,X)-R|_+} ,\label{eq: TF1}
\end{flalign}
where $\calT_{d,q,P}(\Delta)$ is defined in (\ref{eq: cal T alpha definition}).

\begin{lemma} \label{lem:continuity}
	Consider a DMC $W$ and an input distribution $P \in \calP(\calX)$, along with continuous and bounded $d,q$ and a threshold $\Delta$.  For any sequence $P_n \in \calP(\calX)$ with the same support as $P$ such that $P_n(x) \to P(x)$ for all $x$, we have
	\begin{equation}
	\liminf_{n \to \infty} E_{\mathrm{RGV}}(R,P_n,W,q,d,\Delta) \ge E_{\mathrm{RGV}}(R,P,W,q,d,\Delta). \label{eq:cont_LB}
	\end{equation}
\end{lemma}
\begin{proof}
	See Appendix \ref{sec:pf_continuity}.
\end{proof}

\subsection{Proof of Ensemble Tightness (Lower Bound on $\bar{P}_\e^{(n)}$)}\label{sc: lower bound}

We proceed in two steps.

\vspace{0.2cm}
\noindent {\em \underline{Step 1: Lower bounding the conditional error probability }}
\vspace{0.2cm}

We shall use the de Caen's lower bound on the probability of a union \cite{de1997lower} of events $\{\calE_i\}_{i=1}^M$:
\begin{flalign}\label{eq: DeCaenInq-original}
\Pr\Bigg(\bigcup_{i=1}^{M} \calE_i
 \Bigg) \geq  \sum_{i=1}^{M}\frac{\left[\Pr(\calE_i)\right]^2}{\sum_{j=1}^M \Pr(\calE_i\cap\calE_j
 ) }.
\end{flalign}
Explicitly taking into account the case in which $\Pr(\calE_i)$ can be zero for some $i$ values, in which case $\Pr\big(\bigcup_{i=1}^{M} \calE_i
 \big)=\Pr\big(\bigcup_{i \,:\, \Pr(\calE_i)>0} \calE_i
 \big)$, we rewrite the de Caen bound as follows:
\begin{flalign}\label{eq: DeCaenInq_rewrite}
\Pr\Bigg(\bigcup_{i=1}^{M} \calE_i
 \Bigg) \geq  \sum_{ \substack{i=1, \\ \Pr(\calE_i)>0}}^{M}\frac{\left[\Pr(\calE_i)\right]^2}{\Pr(\calE_i)+\sum_{j \ne i}\Pr(\calE_i\cap\calE_j
 ) }.
\end{flalign}
Recalling \eqref{eq: avg err prob}-\eqref{eq:pem2}, and applying \eqref{eq: DeCaenInq_rewrite} to the events $\{\calE_i\}_{i=1}^{M_n}$ defined in \eqref{eq:E_i},
we obtain 
\begin{flalign}\label{eq: DeCaenInq}
&\Pr\Bigg(\bigcup_{\substack{i=1\\ i\neq m}}^{M_n} \calE_i
 \,\bigg|\, \bX_m = \bx_m,\bY=\by\Bigg) \nonumber \\
 \geq&  \sum_{\substack{i=1, i\neq m,\\ \Pr(\calE_i|\bx_m,\by)>0
 }}^{M_n}\frac{\left[\Pr(\calE_i|\bx_m,\by)\right]^2}{\Pr(\calE_i|\bx_m,\by)+\sum_{j\notin \{i,m\}}\Pr(\calE_i\cap\calE_j|\bx_m,\by) }.
\end{flalign}

We first lower bound $ \Pr(\calE_i | \bx_m,\by) $ using \eqref{eq:case_1 rev}
along with Lemmas \ref{lm: pairwise distribution of codewords} and \ref{lm: marginal distribution Xm lemma} to obtain
\begin{flalign}\label{eq: lower bound on Pr Cal E}
 &\Pr(\calE_i | \bx_m,\by) \nonumber \\
   &\geq (1-4\delta_n^2)\e^{-2\delta_n} \sum_{\substack{
    \bx' \in \calT(P_n) \,:\,  q^{}(\bx',\by)\geq q^{}(\bx_m,\by) \\d(\bx',\bx_m)> \Delta}}\frac{1}{|\calT(P_n)|} .
\end{flalign}
Next, we evaluate $\Pr(\calE_i\cap\calE_j|\bx_m,\by)$.
To this end we let $\calI_{d,\Delta}(\bx_i,\bx_j,\bx_m)$ denote the indicator of the event of $(\bx_i,\bx_j,\bx_m)$ mutually satisfying the pairwise $d$-distance constraints; i.e., 
\begin{flalign}\label{eq: indicator triple}
&\calI_{d,\Delta}(\bx_i,\bx_j,\bx_m) \nonumber\\
&=\indicator\{\min\{d(\bx_i,\bx_j),d(\bx_i,\bx_m),d(\bx_m,\bx_j)\}>\Delta\}.\end{flalign}
From Lemmas \ref{lm: triplets-wise distribution of codewords} and \ref{lm: marginal distribution Xm lemma}, we obtain
\begin{flalign}\label{eq: triple upper bound 3}
\Pr(\bx_i,\bx_j|\bx_m)&=\frac{\Pr(\bx_i,\bx_j,\bx_m)}{\Pr(\bx_m)}\\
&=\frac{\Pr(\bx_i,\bx_m,\bx_j)}{1/|\calT(P_n)|}\label{eq: app lemma 7}\\
&\le \frac{\calI_{d,\Delta}(\bx_i,\bx_j,\bx_m)}{(1-e^{-n\delta})^3|\calT(P_n)|^2}\label{eq: app lemma 8}.
\end{flalign}
Now, since $(\bX_i,\bX_j)-\bX_m-\bY$ forms a Markov chain,  
\begin{flalign}
&\Pr(\calE_i\cap\calE_j|\bx_m,\by) \nonumber \\
&=\sum_{\substack{\bx_i,\bx_j:\; q(\bx_i,\by)\geq q(\bx_m,\by),\\ q(\bx_j,\by)\geq q(\bx_m,\by)}}\Pr(\bx_i,\bx_j|\bx_m,\by) \\
&=\sum_{\substack{\bx_i,\bx_j:\; q(\bx_i,\by)\geq q(\bx_m,\by), \\ q(\bx_j,\by)\geq q(\bx_m,\by)}}\Pr(\bx_i,\bx_j|\bx_m)\\
&\dotleq \sum_{\substack{\bx_i,\bx_j \in \calT(P_n) \,:\, q(\bx_i,\by)\geq q(\bx_m,\by), \\q(\bx_j,\by)\geq q(\bx_m,\by)}}\frac{\calI_{d,\Delta}(\bx_i,\bx_j,\bx_m)}{|\calT(P_n)|^2}\label{eq: triple upper bound 2}\\
&\leq \sum_{\substack{\bx_i,\bx_j  \in \calT(P_n) \,:\\ q(\bx_i,\by)\geq q(\bx_m,\by), \\q(\bx_j,\by)\geq q(\bx_m,\by)}}\frac{
\indicator\{\min\{d(\bx_i,\bx_m),d(\bx_m,\bx_j)\}>\Delta\}
}{|\calT(P_n)|^2}\label{eq: triple indicator split}\\
&= \sum_{\bx_i  \in \calT(P_n) \,:\, q(\bx_i,\by)\geq q(\bx_m,\by)}\frac{\indicator\{d(\bx_i,\bx_m)>\Delta \}}{|\calT(P_n)|}\nonumber\\
&\times \sum_{\bx_j  \in \calT(P_n) \,:\,  q(\bx_j,\by)\geq q(\bx_m,\by)} \frac{\indicator\{d(\bx_j,\bx_m)>\Delta\}}{|\calT(P_n)|}\\
&\dotleq \Pr(\calE_i|\bx_m,\by)\Pr(\calE_j|\bx_m,\by)\label{eq: plug in caen}.
\end{flalign}
where \eqref{eq: triple upper bound 2} follows from \eqref{eq: app lemma 8}, \eqref{eq: triple indicator split} follows since by definition of $\calI_{d,\Delta}(\bx_i,\bx_j,\bx_m)$ (see \eqref{eq: indicator triple}), and \eqref{eq: plug in caen} follows from \eqref{eq: lower bound on Pr Cal E}.

Combining \eqref{eq: DeCaenInq} and \eqref{eq: plug in caen} yields
\begin{flalign}
&\Pr\Biggl(\bigcup_{\substack{i=1, \\ i\neq m}}^{M_n}\calE_i|\bx_m,\by\Biggr) \nonumber \\
&\dotgeq  \sum_{\substack{i=1, i\neq m, \\  \Pr(\calE_i|\bx_m,\by)>0}}^{M_n} \left[\Pr(\calE_i|\bx_m,\by)\right]^2 \bigg( \Pr(\calE_i|\bx_m,\by) \nonumber \\
    &\qquad\quad+\Pr(\calE_i|\bx_m,\by)\cdot \sum_{j\notin \{i,m\}}\Pr(\calE_j|\bx_m,\by) \bigg)^{-1} \\
&=  \sum_{\substack{i=1,\\ i\neq m}}^{M_n}\frac{\Pr(\calE_i|\bx_m,\by)}{1+ \sum_{j\notin \{i,m\}}\Pr(\calE_j|\bx_m,\by) },\label{eq: jdsagfkhghg}
\end{flalign}
 where \eqref{eq: jdsagfkhghg} follows since fixing $i$ we have that
if $\Pr(\calE_i|\bx_m,\by)=0$, then obviously the $i$-th summand on the r.h.s.\ of \eqref{eq: jdsagfkhghg} is equal to zero and therefore does not affect the summation, and if 
$\Pr(\calE_i|\bx_m,\by)>0$, this term can be cancelled out from both the numerator and denominator, in which case the $i$ summand on the l.h.s.\ of \eqref{eq: jdsagfkhghg} is equal to that of the r.h.s.\ of \eqref{eq: jdsagfkhghg}. 

Since \eqref{eq:case_5 rev} and \eqref{eq: lower bound on Pr Cal E} imply that 
\begin{flalign}
\Pr(\calE_i|\bx_m,\by) \doteq \sum_{\substack{
    \bx' \,:\,  q^{}(\bx'\by)\geq q^{}(\bx_m,\by) \\d(\bx',\bx_m)> \Delta}}\frac{1}{|\calT(P_n)|},\label{eq: doteq fhdkh}
\end{flalign}
letting $\tilde{p}(\bx_m,\by)$ denote the right-hand side of \eqref{eq: doteq fhdkh}, we obtain
\begin{flalign}
    \Pr\Biggl(\bigcup_{\substack{i=1, \\ i\neq m}}^{M_n}\calE_i|\bx_m,\by\Biggr) & \dotgeq
    (M_n-1)\cdot \frac{\tilde{p}(\bx_m,\by)}{1+(M_n-2)\tilde{p}(\bx_m,\by)}\\
    &\geq 
     \frac{(M_n-1)\tilde{p}(\bx_m,\by)}{1+(M_n-1)\tilde{p}(\bx_m,\by)}\\
    &\geq  \frac{1}{2}\min\{1,(M_n-1)\tilde{p}(\bx_m,\by)\},
\end{flalign}
where the last step follows from the inequality $\frac{x}{1+x}
\geq  \frac{1}{2}\min\{1,x\}$, which holds for all $x\geq 0$.

Averaging over $(\bx_m,\by)$ via Lemma \ref{lm: marginal distribution Xm lemma}, and substituting the definition of $\tilde{p}(\bx_m,\by)$, we obtain the lower bound 
\begin{flalign}\label{eq: first summand final final}
    \bar P_{e,m}^{(n)}&\dotgeq  \sum_{\bx \in \calT(P_n),\by }\frac{1}{|\calT(P_n)|}W^n(\by|\bx) \nonumber\\
    &\times \min\Bigg\{1, (M_n-1)\sum_{\substack{\bx' \in \calT(P_n) \,:\, q^{}(\bx',\by)\geq q^{}(\bx,\by)\\ d(\bx',\bx)> \Delta }} \frac{1}{|\calT(P_n)|} \Bigg\}.
\end{flalign}

\vspace{0.2cm}
\noindent {\em \underline{Step 2: Deducing the error exponent }}
\vspace{0.2cm}

\noindent Applying a similar argument to that used in deriving (\ref{eq: TF1}), we obtain from \eqref{eq: first summand final final} that
\begin{multline}\label{eq: standard method of type 2}
    \bar{P}_\e^{(n)} \dotgeq \exp \Big(-n\min_{V\in\calT^{(n)}_{d,q}(\Delta)}D(V_{Y|X}\|W|P_n) \\ +\big|I(\widetilde{X};Y,X)-R\big|_+\Big)
 \end{multline} 
 where 
\begin{flalign}\label{eq: cal T definition star n}
   &\calT_{d,q,P}^{(n)}(\Delta)\triangleq \Bigl\{
   V_{X\widetilde{X}Y}\in\mathcal{P}_n(\calX\times\calX\times\calY) \,:\,\nonumber \\ &V_X=V_{\widetilde{X}}=P_n, q(V_{\widetilde{X}Y})\geq q(V_{XY}),\,d(P_{X\widetilde{X}})\geq \Delta
    \Bigr\}.
\end{flalign}
Note that this exponent differs from $E_{\rm RGV}(R,P,W,\phi,d, \Delta)$ only in that the minimization is performed over empirical distributions rather than the probability simplex. 
The following lemma concludes the proof of ensemble tightness; this is the only part of the analysis where the assumption of additive $q$ is used.
\begin{lemma}\label{lm: empirical simplex lemma}
    Given $P\in \mathcal{P}(\calX)$ and its corresponding type $P_n\in \mathcal{P}_n(\calX)$, under any $d \in \Omega$ and additive and bounded metric $q$, we have for any $\epsilon>0$ and sufficiently large $n$ that
    \begin{flalign}\label{eq: standard method of type 2 lemma}
        &\min_{V\in\calT^{(n)}_{d,q}(\Delta)}D(V_{Y|X}\|W|P_n)+\big|I(\widetilde{X};Y,X)-R\big|_+ \nonumber\\
        &\qquad \leq E_{\mathrm{RGV}}(R,P,W,q,d,\Delta+\epsilon) + \epsilon.
     \end{flalign}
\end{lemma}
 \noindent The proof of Lemma \ref{lm: empirical simplex lemma} is given in Appendix \ref{ap: empirical simplex lemma appendix}.

\section{Dual Expression and Continuous Alphabets} \label{sc:dual}

In this section, we show that in the case that the distance function $d$ and decoding metric $q$ are additive, the RGV exponent of Theorem \ref{th: main theorem sphere packing} permits an equivalent dual expression obtained using Lagrange duality. Moreover, we explain how it can be derived directly using cost-constrained coding \cite{Gallager68,ScarlettMartinezGuilleniFabregas_mismatch_2014_IT,ScarlettPengMerhavMartinezGuilleniFabregas_mismatch_2014_IT}, without resorting to constant-composition coding.  This approach extends directly to memoryless channels with infinite or continuous alphabets under mild technical assumptions, namely, that all auxiliary cost functions involved have a finite mean with respect to $P$.

\subsection{Dual expression}

We begin by stating the dual form of the RGV exponent and rate condition in Theorem \ref{th: main theorem sphere packing}.  As mentioned above, we focus on additive distances of the form \eqref{eq:additive_d}, and additive decoding metrics of the form \eqref{eq:q_additive}

\begin{theorem} \label{th:Lagrange}
    Under the setup of Theorem \ref{th: main theorem sphere packing} with an additive distance function $d$ and additive decoding metric $q$, the error exponent \eqref{eq: E_ex dfn} can be written as
    \begin{flalign}
        &E_{\mathrm{RGV}}(R,P,W,q,d, \Delta) \nonumber \\ 
        &= \sup_{\rho\in[0,1],r\ge0,s\ge0,a(\cdot)}-\sum_{x}P(x)\log\sum_{y}W(y|x) \nonumber \\
        &~\times\bigg(\frac{\sum_{x'}P(x')e^{sq(x',y)}e^{a(x')}e^{r(d(x,x')-\Delta)}}{e^{sq(x,y)}e^{a(x)}}\bigg)^{\rho}-\rho R, \label{eq:E_dual}
    \end{flalign}
    and rate condition \eqref{eq: frist eq} can be written as
    \begin{flalign}
        &R \le \sup_{\substack{ r\ge0, a(\cdot)}} -\sum_{x}P(x) \nonumber\\
        & \quad\times \log\sum_{x'}P(x')e^{a(x')-\phi_a}e^{-r(d(x,x')-\Delta)}-2\delta, \label{eq:R_dual}
    \end{flalign}
    where $\phi_a = \EE_P[a(X)]$.
\end{theorem}
\begin{proof}
    The proof uses Lagrange duality analogously to the corresponding statements for the random coding and expurgated exponents \cite{ScarlettMartinezGuilleniFabregas_mismatch_2014_IT,ScarlettPengMerhavMartinezGuilleniFabregas_mismatch_2014_IT}; see Appendix \ref{sc:Lagrange}.
\end{proof}

The expression in \eqref{eq:E_dual} bears a strong resemblance to the mismatched random coding exponent for constant-composition coding \cite{ScarlettMartinezGuilleniFabregas2012AllertonSU}; in fact, the only difference is the presence of additional term $e^{r(d(x,x')-\Delta)}$.

The proof of Theorem \ref{th:Lagrange} does not use the symmetry of $d$, and hence the equivalence holds even for non-symmetric $d$ as per Corollary \ref{cor:asymm}.  The direct derivation below, however, does require a symmetric distance function, but one can still infer the achievability of the exponent for non-symmetric choices via the symmetrization argument used in Corollary \ref{cor:asymm}.

\subsection{Direct derivation via cost-constrained coding} \label{sec:direct}

One way of understanding \eqref{eq:E_dual} is by noting that it is the exponent that one obtains upon applying Gallager-type bounding techniques, e.g., Markov's inequality and $\min\{1,\alpha\} \le \min_{\rho\in[0,1]} \alpha^\rho$, to the asymptotic multi-letter random coding union bound expression in \eqref{eq:rcu-like} for constant-composition coding.  To our knowledge, the ``dual analysis'' of constant-composition random coding was initiated by Poltyrev \cite{PoltyrevRandom}.

The preceding approach permits continuous channel outputs, but requires discrete inputs.  It turns out, however, that we can attain an analog of \eqref{eq:rcu-like} for a cost-constrained coding scheme in which the input may also be continuous.  In this section, we describe the changes needed in the code construction and analysis for this purpose.  To simplify the presentation, we still use summations to denote averaging, but these can directly be replaced by integrals in continuous-alphabet settings.  A disadvantage of  this approach is that it is difficult to claim ensemble tightness; we provide only achievability results.

\subsubsection{Code construction}
Fix an input distribution $P$ and four auxiliary costs $a_{1}(x),\dotsc,a_{4}(x)$. Let $P^{n}$ be the $n$-fold product of $P$, let $a_{j}(\bx) = \frac{1}{n}\sum_{k=1}^n a_j(x_k)$ be the normalized additive extension of $a_j$, and define the cost-constrained distribution 
\begin{equation}
    P_{\bX}(\bx)=\frac{1}{\mu}P^{n}(\bx)\indicator\Big\{ \big| a_{j}(\bx) - \phi_{j}\big| \le \epsilon, ~~j=1,2,3,4 \Big\},\label{eq:cost}
\end{equation}
where $P^n(\bx) = \prod_{k=1}^n P(x_k)$, $\phi_j = \EE_P[a_j(X)]$, $\epsilon > 0$ is a parameter, and $\mu$ is a normalizing constant.  Note that the functions $a_j$ represent {\em auxiliary costs} that are intentionally introduced to improve the performance (in terms of the error exponent) of the random-coding ensemble.  One can incorporate a {\em system cost} (e.g., a power constraint) in exactly the same way to ensure a per-codeword constraint of the form $\frac{1}{n} \sum_{k=1}^n c(x_k) \le \Gamma$ for some cost function $c$ and threshold $\Gamma$; in such cases (which are crucial for continuous-alphabet settings), all of the subsequent analysis remains unchanged as long as $P$ is chosen to satisfy $\EE_P[c(X)] < \Gamma$. 

By definition, $P_{\bX}$ is i.i.d.~conditioned on each $a_j$ being close to its mean.  We observe that $\mu$ is the probability (under $P^n$) of the event in the indicator function of \eqref{eq:cost} occurring, and we immediately obtain
\begin{equation}
    \lim_{n\to\infty} \mu = 1
\end{equation}
by the law of large numbers.

With the definition of $P_{\bX}$ in place, we recursively generate the codewords in a similar manner to Section \ref{sc: Coding Scheme}:
\begin{gather}
    \Pr(\bx_{1})=P_{\bX}(\bx_{1}) \label{eq:p_cost_1} \\
    \Pr(\bx_{2}|\bx_1)=\frac{1}{\mu_{2}(\bx_1)}P_{\bX}(\bx_{2})\indicator\big\{ d(\bx_{1},\bx_{2})>\Delta\big\}\\
    \vdots \nonumber
\end{gather}
\vspace*{-2ex}
\begin{multline}
    \Pr(\bx_{m}|\bx_1^{m-1})
    =\frac{1}{\mu_{m}(\bx_1^{m-1})}P_{\bX}(\bx_{m}) \\ \times \indicator\big\{ d(\bx_{i},\bx_{m})>\Delta\,\,\,\forall i<m\big\}, \label{eq:p_cost_m} 
\end{multline}
where each $\mu_m(\cdot)$ is a normalizing constant depending on all of the previous codewords.  Note that in the case of continuous alphabets, each probability $\Pr(\bx_i \,|,\ \cdot)$ should be replaced by a conditional density function $f(\bx_i \,|\, \cdot)$.

We proceed by describing the analysis in two steps. To avoid repetition, we omit certain parts of the analysis that are the same as the constant-composition case.

\subsubsection{Key properties} 

Similarly to the constant-composition case, we seek to arrive at an upper bound of the form 
\begin{flalign}
&\hspace{-2mm}\bar{P}_\e^{(n)}\dotleq \sum_{\bx,\by} P_{\bX}(\bx) W^n(\by|\bx) \nonumber 
\\
&\times \min\Biggl\{1,(M_n-1) \sum_{\substack{\bx' \,:\, q^{}(\bx',\by)\geq q^{}(\bx,\by)\\ d(\bx',\bx)\geq \Delta }}  P_{\bX}(\bx') \Biggr\} \label{eq:rcu-like-2}
\end{flalign}
that holds under the rate condition \eqref{eq:R_dual}.  Towards establishing this bound, we prove the following four important properties:
\begin{enumerate}[~\em (a)]
	\item For any $\bx$ such that $P_{\bX}(\bx) > 0$, we have under $\bX' \sim P_{\bX}$ that
	\begin{multline}
		-\frac{1}{n}\log\Pr\big(d(\bx,\bX')\le\Delta\big) 
		\ge \sup_{r\ge0,a(\cdot)}-\sum_{x}P(x) \\
		\times \log\sum_{x'}P(x')e^{a(x')-\phi_a}e^{-r(d(x,x')-\Delta)}-\delta,\label{eq:final0}
	\end{multline}
	thus matching the rate condition in \eqref{eq:R_dual}.
    \item The normalizing constants in \eqref{eq:p_cost_1}--\eqref{eq:p_cost_m} satisfy $\mu_{m}(\bX_1^{m-1}) \ge 1-\e^{-n\delta}$ almost surely under the rate condition \eqref{eq:R_dual}, for any choice of $\delta > 0$.
    \item The marginal distribution of any given codeword (indexed by $m$) satisfies $\Pr(\bx_m) \doteq P_{\bX}(\bx_m)$.
    \item The marginal joint distribution of any two codewords (indexed by $k$ and $m$) satisfies $\Pr(\bx_k,\bx_m) \dotleq P_{\bX}(\bx_k)P_{\bX}(\bx_m)\indicator\{ d(\bx_k,\bx_m) > \Delta \}.$
\end{enumerate}
As we outline below, the first two properties are used as stepping stones to obtaining the final two.  Once the final two properties are established, then a near-identical analysis to that of \eqref{eq:E_i}--\eqref{eq:final_rcu} yields \eqref{eq:rcu-like-2}.

To establish the first property \eqref{eq:final0}, we bound the probability therein for fixed $\bx$:
\begin{align}
    &\Pr\big(d(\bx,\bX')\le\Delta\big) \nonumber \\
    & =\sum_{\bx'}P_{\bX}(\bx')\indicator\big\{ d(\bx,\bx')\le \Delta\big\} \label{eq:Markov_d}\\
    & \le\sum_{\bx'}P_{\bX}(\bx')\e^{-nr(d(\bx,\bx')-\Delta)} \label{eq:Markov_d2}\\
    & \le\sum_{\bx'}P_{\bX}(\bx')\e^{-nr(d(\bx,\bx')-\Delta)}\e^{n(a_{1}(\bx')-\phi_{1} + \epsilon)} \label{eq:Markov_d3}\\
    & \le\sum_{\bx'}P^{n}(\bx')\e^{-nr(d(\bx,\bx')-\Delta)}\e^{n(a_{1}(\bx')-\phi_{1} + 2\epsilon)}, \label{eq:Markov_d4}
\end{align}
where \eqref{eq:Markov_d2} uses Markov inequality with an arbitrary parameter $r > 0$, \eqref{eq:Markov_d3} uses the fact that $a_{1}(\bx') \ge \phi_{1} - \epsilon$ by construction, and \eqref{eq:Markov_d4} holds for sufficiently large $n$ because $\mu\to1$ in (\ref{eq:cost}). Taking the logarithm and applying Gallager's single-letterization argument \cite{Gallager68}, we get 
\begin{flalign}
    &-\log\Pr\big(d(\bx,\bX')\le\Delta\big) \nonumber \\
    &~~\ge-\sum_{k=1}^{n}\log\sum_{x'}P(x')\e^{-r(d(x_{k},x')-\Delta)}\e^{a_{1}(x')-\phi_{1}} - 2n\epsilon. \label{wq:blah}
\end{flalign}
We now choose $a_{2}(x)=-\log\sum_{x'}P(x')\e^{r(d(x,x')-\Delta)}\e^{a_{1}(x')-\phi_{1}}$,
which ensures that the leading term on the right-hand side of (\ref{wq:blah}) is equal to
$n a_{2}^{n}(\bx)$.  Hence, substituting the definition $\phi_{2}=\EE_{P}[a_{2}(X)]$ and using $a_{2}(\bx) \ge \phi_{2} - \epsilon$ by construction, we obtain
\begin{flalign}
    &-\frac{1}{n}\log\Pr\big(d(\bx,\bX')\le\Delta\big) \nonumber \\
    &\,\ge -\sum_{x}P(x)\log\sum_{x'}P(x')\e^{-r(d(x,x')-\Delta)}\e^{a_{1}(x')-\phi_{1}} - 3\epsilon.\label{eq:final}
\end{flalign}
Choosing $\epsilon = \frac{\delta}{3}$ and optimizing $r$ and $a_1(\cdot)$, we obtain \eqref{eq:final0}, thus completing the proof of the first property above.

The second property above follows easily from the first: Letting $\bX' \sim P_{\bX}$, we have $\mu_{m}(\bx_1^{m-1})=\Pr(d(\bx_{i},\bX')>\Delta,\,\forall i < m)$, and the union bound gives 
\begin{align}
    1-\mu_{m}(\bx_1^{m-1})
    &\le \sum_{i < m}\Pr\big(d(\bx_{i},\bX')\le \Delta\big)\label{eq:1-mu} \\
    &\le e^{nR_n} \Pr\big(d(\bx_{i},\bX')\le \Delta\big)\\
    &\le e^{-n\delta}, \label{eq:second_cond_3}
\end{align}
where \eqref{eq:second_cond_3} follows from \eqref{eq:final0} and the rate condition \eqref{eq:R_dual}. 

Upper bounding the indicator functions in \eqref{eq:p_cost_1}--\eqref{eq:p_cost_m} by one gives $\Pr(\bx_m) \dotleq P_{\bX}(\bx_m)$, thus proving one direction of the dot-equality in the third property above.  The other direction requires more effort, and is deferred to Appendix \ref{app:cost_lb}.

For the forth property above, we use \eqref{eq:p_cost_m} and the fact that $\mu_m(\bx_1^{m-1}) \ge 1-e^{-n\delta}$ to obtain
\begin{flalign}
&\Pr(\bx_k,\bx_m) \nonumber\\
&\,= \sum_{\bx_1^{k-1}, \bx_{k+1}^{m-1}} \Pr(\bx_1^{k-1}) \Pr(\bx_k|\bx_1^{k-1})\nonumber\\
&\quad\quad\times \Pr(\bx_{k+1}^{m-1}|\bx_1^k) \Pr(\bx_{m}|\bx_{1}^{m-1}) \\
&\,\le \sum_{\bx_1^{k-1}, \bx_{k+1}^{m-1}} \Pr(\bx_1^{k-1}) \cdot \frac{P_{\bX}(\bx_k)}{1-e^{-n\delta}} \cdot \Pr(\bx_{k+1}^{m-1}|\bx_1^k) \nonumber\\
&\quad \quad \times  \frac{P_{\bX}(\bx_m) \indicator\{ d(\bx_k,\bx_m) > \Delta \}}{1-e^{-n\delta}} \\
&\,=  \frac{1}{(1-e^{-n\delta})^2} P_{\bX}(\bx_k) P_{\bX}(\bx_m)\indicator\{ d(\bx_k,\bx_m) > \Delta \}.
\end{flalign}

\subsubsection{Upper-bounding the multi-letter upper bound}

Once \eqref{eq:rcu-like-2} is established, the steps in deriving \eqref{eq:E_dual} are standard.  Such an analysis requires two additional auxiliary costs, and these are given by $a_3$ and $a_4$ in \eqref{eq:cost}. In particular, we set $a_3(x) = a(x)$ in \eqref{eq:E_dual} and
\begin{flalign}
&a_4(x) = -\log\sum_{y}W(y|x)\nonumber\\
&\times\bigg(\frac{\sum_{x'}P(x')e^{sq(x',y)}e^{a(x')}e^{r(d(x,x')-\Delta)}}{e^{sq(x,y)}e^{a(x)}}\bigg)^{\rho}.
\end{flalign}

In fact, removing the constraint $d(\bx,\bx')>\Delta$ from the pairwise error probability term in \eqref{eq:rcu-like-2} recovers the standard random-coding union bound, which was already used in \cite{ScarlettMartinezGuilleniFabregas_mismatch_2014_IT} to establish the exponent in \eqref{eq:E_dual} without the term $\e^{r(d(x,x') - \Delta)}$.  Hence, the change in the analysis compared to \cite{ScarlettMartinezGuilleniFabregas_mismatch_2014_IT} only amounts to an application of the inequality $\indicator\{ d(\bx,\bx') \ge \Delta \} \le \e^{nr(d(\bx,\bx') - \Delta)}$, similarly to \eqref{eq:Markov_d}.  Due to this similarity, the details are omitted.

\section{Optimal Distance Functions}\label{sc: Optimal Choices for d}

In this section, we study optimal choices for the distance function $d(\cdot,\cdot)$ in Theorem \ref{th: main theorem sphere packing}, thus characterizing how the codewords should be separated in order to get the best possible exponent for our construction. While some of the analysis in this section includes the constant $\delta>0$, the best exponent will always be obtained in the limit as $\delta \to 0$.

\subsection{Reduction to the Csisz\'ar-K\"orner Exponent - Negative Mutual Information Distance}
We show that when the distance function $d(\cdot,\cdot)$ is optimized, and $\Delta$ is chosen appropriately, the exponent in Theorem \ref{th: main theorem sphere packing} recovers the exponent $E_{q}(R,P,W)$ in (\ref{eq: E_ex dfn C and K}) \cite{CsiszarKorner81graph}.

\begin{corollary} \label{cor:ck_exponent}
     Let $\epsilon>0$ be given, let $q(\cdot)$ be an arbitrary type-dependent continuous decoding rule, and let $R$, $P$, and $d \in \Omega$ be given. The exponent of the ensemble average error probability of the generalized RGV construction with sufficiently small $\delta$,  $d(P_{X\widetilde{X}})=-I(X;\widetilde{X})$, $\Delta=-(R+2\delta)$,  sufficiently large $n$, and decoding metric $q(\cdot)$ over the DMC $W$ is at least as high as $E_{q}(R,P,W)-\epsilon$.

   \end{corollary}

\begin{proof}
We claim that the choices $d(P_{X\widetilde{X}})=-I(X;\widetilde{X})$ and $\Delta=-(R+2\delta)$ are valid for all $R$ in the sense of satisfying the rate condition \eqref{eq: frist eq}.  To see this, note that
\begin{flalign}
&\min_{\substack{P_{X\widetilde{X}} \,:\, d(P_{X\widetilde{X}}) \leq\Delta\\ P_X=P_{\widetilde{X}}=P}} I(X;\widetilde{X})\bigg|_{\substack{d(P_{X\widetilde{X}})=-I(X;\widetilde{X})\\ \Delta=-(R+2\delta)}} \nonumber  \\
&\quad =\min_{\substack{P_{X\widetilde{X}} \,:\, I(X;\widetilde{X})\geq R+2\delta
\\ P_X=P_{\widetilde{X}}=P}} I(X;\widetilde{X})\\
&\quad \geq R+2\delta,
\end{flalign}
as required.
Now, under the same choices, we have
    \begin{flalign}\label{eq: rhs close to EqR}
        & E_{\mathrm{RGV}}(R,P,W,q,d, \Delta)\Big|_{\substack{d(P_{X\widetilde{X}})=-I(X;\widetilde{X}),\,  \Delta=-(R+2\delta)}}\nonumber\\
        &\quad = \min_{V\in\calT_{I,\delta} }D(V_{Y|X}\|W|P)+\big|I_{}(\widetilde{X};Y,X)-R\big|_+,
     \end{flalign}
    where
    \begin{flalign}
        \calT_{I,\delta} &\triangleq \Big\{\begin{array}{l}
        V_{X\widetilde{X}Y} \in \mathcal{P}(\calX\times\calX\times\calY) \,:\, V_X=V_{\widetilde{X}}=P,\\
        q(V_{\widetilde{X}Y})\geq q(V_{XY}),
         I_{}(\widetilde{X};X)\leq R+3\delta
         \end{array}\Big\}. \label{eq:set_TI}
    \end{flalign}
    The result follows by taking $\delta \to 0$ and using the continuity of $E_{q}(R,P,W)$ in $R$ \cite{CsiszarKorner81graph}. 
 \end{proof}

The following proposition reveals that the above choice of $(d,\Delta)$ is the one that maximizes the general exponent given in Theorem \ref{th: main theorem sphere packing}.

\begin{proposition}\label{th: alpha optimality}
    Under the setup of Theorem \ref{th: main theorem sphere packing} with 
    \beq\label{eq: 45 again}
    R\leq \min_{\substack{P_{X\widetilde{X}} \,:\, P_X=P_{\widetilde{X}}=P\\d(P_{X\widetilde{X}}) \leq \Delta}} I(X;\widetilde{X})-2\delta,
    \eeq
    we have
    \begin{flalign}
    &E_{\mathrm{RGV}}(R,P,W,q,d, \Delta) \nonumber \\
    &\leq E_{\mathrm{RGV}}(R,P,W,q,d, \Delta)\Big|_{d=-I(X;\widetilde{X}) , \;\Delta=-(R+2\delta)}.
    \end{flalign}
\end{proposition}
\begin{proof}
    From \eqref{eq: 45 again}, we see that among all $P'_{X\widetilde{X}}$ such that  $P'_X=P'_{\widetilde{X}}=P$, the condition $d(P'_{X\widetilde{X}}) \leq \Delta$ implies $R+2\delta\leq I_{P'}(X;\widetilde{X})$. The contrapositive statement is that among all $P_{X\widetilde{X}}'$   such that  $P'_X=P'_{\widetilde{X}}=P$, the condition $R+2\delta > I_{P'}(X;\widetilde{X})$ implies $d(P_{X\widetilde{X}}') > \Delta$.
    As a result, when \eqref{eq: 45 again} holds, $\calT_{d,q,P}(\Delta)$ defined in (\ref{eq: cal T alpha definition}) satisfies 
    \begin{flalign}
        \calT_{d,q,P}(\Delta)\supseteq\calT_{I,\delta},
    \end{flalign}
    where $\calT_{I,\delta}$ is defined in \eqref{eq:set_TI}.
    Therefore,
    \begin{flalign}\label{eq: optimal choice proposition}
       &\hspace{-1mm} E_{\mathrm{RGV}}(R,P,W,q,d, \Delta) \nonumber\\
        &= \min_{V\in\calT_{d,q,P}(\Delta) }D(V_{Y|X}\|W|P)+\big|I_{}(\widetilde{X};Y,X)-R\big|_+\\
        &\leq ~\min_{V\in\calT_{I,\delta} }~~D(V_{Y|X}\|W|P)+\big|I_{}(\widetilde{X};Y,X)-R\big|_+,
    \end{flalign}
    so the exponent is upper bounded by that corresponding to $d(P_{X\widetilde{X}})= -I(X;\widetilde{X})$ and $\Delta=-(R+2\delta)$.    
\end{proof}

We note that the choice $d(P_{X\widetilde{X}})=-I(X;\widetilde{X})$ is {\em universally optimal} in maximizing the achievable exponent in Theorem \ref{th: main theorem sphere packing} (subject to \eqref{eq: frist eq}), in the sense that it has no dependence on the channel, decoding rule, or input distribution. This provides an interesting analogy with the decoding rule $q(P_{XY}) = I(X;Y)$, which is known to be universally optimal for achieving the regular random-coding exponent; however, it remains an open problem as to whether such a choice also attains the expurgated exponent \cite{CsiszarKorner81graph}.

\subsection{A Non-Universal Optimal Distance Function} \label{sec:non_univ}

In this subsection, we show that the non-universal distance function $d(P_{X\tilde X})=\beta_{R,W,q}(P_{X\widetilde{X}})$ also achieves the exponent of Csisz\'ar and K\"orner, where
\begin{flalign}\label{eq: beta dfn}
     \beta_{R,W,q}(P_{X\widetilde{X}})\triangleq & \min_{V_{X\widetilde{X}Y}\in\calT'(P_{X\widetilde{X}}) }\Gamma(V_{X\tilde X Y}),
\end{flalign}
with\footnote{The dependence of $\Gamma$ on $(R,W,q)$ is left implicit to lighten notation.}
\beq\label{eq: gamma function definition}
\Gamma(V_{X\tilde X Y}) \triangleq D(V_{Y|X}\|W|V_X)+\big|I_{}(\widetilde{X};Y,X)-R\big|_+,
\eeq
and 
\begin{flalign}\label{eq: calT prime dfn}
    \calT'(P_{X\widetilde{X}}) \triangleq  \Bigl\{   V_{X\widetilde{X}Y}\in\mathcal{P}(\calX\times\calX\times\calY) \,:\, \nonumber \\
    V_{X\widetilde{X}}=P_{X\widetilde{X}}, q(V_{\widetilde{X}Y})\geq q(V_{XY}) \Bigr\}.
\end{flalign}
We first provide a corollary characterizing the exponent of Theorem \ref{th: main theorem sphere packing} with $d(\cdot) =  \beta_{R,W,q}(\cdot)$, and then prove its equivalence to (\ref{eq: E_ex dfn C and K}).

\begin{corollary}\label{cr: R Delta region}
    If the pair $(R,\Delta)$ satisfies the condition 
    \begin{flalign}\label{eq: frist eq2}
    R\leq \min_{\substack{P_{X\widetilde{X}} \,:\, P_X=P_{\widetilde{X}}=P\\\beta_{R,W,q}(P_{X\widetilde{X}}) \leq\Delta}} I(X;\widetilde{X})-2\delta,
    \end{flalign}
    then the ensemble average error probability $\bar{P}_\e^{(n)}$ of the RGV codebook construction with parameters $(n,R,P,\beta_{R,W,q}, \Delta,\delta)$ using the continuous type-dependent decoding rule $q(\cdot)$ over the channel $W$ satisfies
    \begin{flalign}
    \bar{P}_\e^{(n)}&\dotleq  %
    \e^{-n\Delta}.
    \end{flalign}
    \end{corollary}
\begin{proof}
    First observe that the minimization in $E_{\mathrm{RGV}}(R,P,W,q, d, \Delta)$ (see (\ref{eq: E_ex dfn})) can be done in two stages: Minimize first over $P_{X\widetilde{X}}$, and then over $V_{X\widetilde{X}Y}$ that are consistent with $P_{X\widetilde{X}}$.  By doing so, we obtain
    \begin{flalign}
        &E_{\mathrm{RGV}}(R,P,W,q,d, \Delta) \nonumber \\
        &= \min_{\substack{P_{X\widetilde{X}} \,:\, P_X=P_{\widetilde{X}}=P\\d(P_{X\widetilde{X}})\geq \Delta}}\, \min_{V_{X\widetilde{X}Y}\in\calT'(P_{X\widetilde{X}}) }D(V_{Y|X}\|W|P) \nonumber \\
        &\hspace*{4cm} +\big|I_{}(\widetilde{X};Y,X)-R\big|_+,
    \label{eq:beta_eq_rgv}
    \end{flalign}
    where $\calT'(P_{X\widetilde{X}}) $ is defined in (\ref{eq: calT prime dfn}). 
    From the definition of $ \beta_{R,W,q}(P_{X\widetilde{X}})$ (\ref{eq: beta dfn}), we can rewrite \eqref{eq:beta_eq_rgv} as
    \begin{flalign}
        E_{\mathrm{RGV}}(R,P,W,q,d, \Delta)&= \min_{\substack{P_{X\widetilde{X}} \,:\, P_X=P_{\widetilde{X}}=P\\d(P_{X\widetilde{X}})\geq \Delta}}\beta_{R,W,q}(P_{X\widetilde{X}}).
    \end{flalign}
    Hence, by the choice $d(\cdot) =\beta_{R,W,q}(\cdot)$ we obtain
    \begin{flalign}
        E_{\mathrm{RGV}}(R,P,W,q,d, \Delta)&= \min_{\substack{P_{X\widetilde{X}} \,:\, P_X=P_{\widetilde{X}}=P\\d(P_{X\widetilde{X}})\geq \Delta}} d(P_{X\widetilde{X}}) \\
        &\geq\Delta.
    \end{flalign}
    Combined with (\ref{eq: frist eq}), this yields that for a pair $(R,\Delta)$ that satisfies (\ref{eq: frist eq2}), we have $\bar{P}_\e^{(n)}\dotleq\e^{-n\Delta}$. 
\end{proof}

Note that while the preceding proof gives an exponent of $\Delta$, one cannot make $\Delta$ arbitrarily large, because past a certain point the condition \eqref{eq: frist eq2} will never be satisfied.

The following proposition shows that error exponents corresponding to Corollaries \ref{cor:ck_exponent} and \ref{cr: R Delta region} are identical, and hence, both are optimal when \eqref{eq: frist eq} holds.

\begin{proposition}
    For any $P \in \mathcal{P}(\calX)$, the achievable rate-exponent pairs $(R,E)$ resulting from Theorem \ref{th: main theorem sphere packing} (i.e., taking the union over all $\delta > 0$ and $\Delta > 0$) are identical for the choices $d(P_{X\widetilde{X}})=-I(X;\widetilde{X})$ and  $d(P_{X\widetilde{X}})=\beta_{R,W,q}(P_{X\widetilde{X}})$.
\end{proposition}
\begin{proof}
    Consider the exponent in Corollary \ref{cor:ck_exponent} for $d(P_{X\widetilde{X}})=-I(X;\widetilde{X})$.  For fixed $R$, the highest possible exponent $E$ is obtained by choosing $\Delta$ such that \eqref{eq: frist eq} holds with equality, and then taking $\delta \to 0$ to obtain the achievable pair
    \begin{flalign}\label{eq: R E I}
        &(R,E) \nonumber \\
        &= \left(R , \min_{
        V_{X\widetilde{X}Y}\in \calT_I
         }D(V_{Y|X}\|W|P)+|I(\widetilde{X};Y,X)-R|_+ \right),
    \end{flalign}
     where
    \begin{flalign}\label{eq: cal T definition}
        &\calT_I \triangleq \Bigl\{
        V_{X\widetilde{X}Y}\in\mathcal{P}(\calX\times\calX\times\calY) \,:\, \nonumber \\
        &V_X=V_{\widetilde{X}}=P, q(P_{\widetilde{X},Y})\geq q(P_{X,Y}), I(X;\widetilde{X})\leq R \Bigr\}.
    \end{flalign}
    
      Next, Corollary \ref{cr: R Delta region} states that $\Delta$ is an achievable exponent at rate $R$ for $d(P_{X\widetilde{X}})=\beta_{R,W,q}(P_{X\widetilde{X}})$ provided that
    \begin{flalign} \label{eq:R_alt}
        R<\min_{P_{X\widetilde{X}} \,:\, P_X = P_{\widetilde{X}} = P, \beta_{R,W,q}(P_{X\widetilde{X}})\leq \Delta}I(X;\widetilde{X}).
    \end{flalign}
The condition $\beta_{R,W,q}(P_{X\widetilde{X}})\leq \Delta$ is equivalent to:
    \begin{align}
        &\beta_{R,W,q}(P_{X\widetilde{X}})\leq \Delta \nonumber \\
        &\iff  \min_{V_{X\widetilde{X}Y} \,:\, q(P_{\widetilde{X},Y})\ge q(P_{X,Y}), V_{X\widetilde{X}}=P_{X\widetilde{X}}}\Gamma(V_{X\widetilde{X}Y})\leq \Delta\\
        &\iff \Gamma(V_{X\widetilde{X}Y})\leq \Delta \text{ for some } V_{X\widetilde{X}Y} \nonumber \\
        & \quad\qquad \text{s.t. } q(P_{\widetilde{X},Y})\ge q(P_{X,Y}), V_{X\widetilde{X}}=P_{X\widetilde{X}}.
    \end{align}
Using this, we can rewrite the right-hand side of \eqref{eq:R_alt} as
    \begin{align}
        &\min_{P_{X\widetilde{X}} \,:\, P_X = P_{\widetilde{X}} = P, \beta_{R,W,q}(P_{X\widetilde{X}})\leq \Delta}I(X;\widetilde{X}) \nonumber \\ 
            &= \min_{\substack{P_{X\widetilde{X}} \,:\, P_X = P_{\widetilde{X}} = P, \\ 
            \Gamma(V_{X\widetilde{X}Y})\leq \Delta \text{ for some } V_{X\widetilde{X}Y} \,:\,
             \substack{q(P_{\widetilde{X},Y})\ge q(P_{X,Y}), \\ V_{X\widetilde{X}}=P_{X\widetilde{X}}} }} I(X;\widetilde{X}) \\
            &= \min_{P_{X\widetilde{X}} \,:\, P_X = P_{\widetilde{X}}=P}\,\, \min_{V_{X\widetilde{X}Y} \,:\, q(P_{\widetilde{X},Y})\ge q(P_{X,Y})}\nonumber\\
                &\quad \quad \begin{cases}
                    I_{P}(X;\widetilde{X}) & \Gamma(V_{X\widetilde{X}Y}) \le \Delta \text{ and } V_{X\widetilde{X}}=P_{X\widetilde{X}} \\
                    \infty & \text{otherwise}
                \end{cases} \\
           &= \min_{\substack{V_{X\widetilde{X}Y} \,:\, q(P_{\widetilde{X},Y})\ge q(P_{X,Y}), \\ \Gamma(V_{X\widetilde{X}Y})\leq \Delta,V_{X}=V_{\widetilde{X}}=P } }I_V(X;\widetilde{X}), \label{eq:iauegiaefk}
    \end{align}
    where the last step uses the fact that $I_{P}(X;\widetilde{X}) = I_{V}(X;\widetilde{X})$ whenever $V_{X\widetilde{X}}=P_{X\widetilde{X}}$.
    From \eqref{eq:iauegiaefk}, it follows that \eqref{eq:R_alt} can be written as 
    \begin{equation}
        R < \min_{V_{X\widetilde{X}Y} \in \calV \,:\,\Gamma(V_{X\widetilde{X}Y})\leq \Delta} I_V(X;\widetilde{X}), \label{eq:equiv1}
    \end{equation}
    where $\calV = \{ V_{X\widetilde{X}Y} \,:\, q(P_{\widetilde{X},Y})\ge q(P_{X,Y}), V_{X}=V_{\widetilde{X}}=P \}$.  We claim that \eqref{eq:equiv1} is equivalent to
    \begin{flalign}
        \Delta < \min_{V_{X\widetilde{X}Y} \in \calV \,:\,  I_V(X;\widetilde{X})\le R}\Gamma(V_{X\widetilde{X}Y}). \label{eq:equiv2}
    \end{flalign}
    To see this, we show that \eqref{eq:equiv1} implies \eqref{eq:equiv2}, and that the complement of \eqref{eq:equiv1} implies the complement of \eqref{eq:equiv2}:
    \begin{itemize}
        \item First suppose that \eqref{eq:equiv1} holds.  This means that within $\calV$ we have $\Gamma(V_{X\widetilde{X}Y}) \le \Delta \implies R < I_V(X;\widetilde{X})$, and the contrapositive statement is that within $\calV$ we have $R \ge I_V(X;\widetilde{X}) \implies \Gamma(V_{X\widetilde{X}Y}) > \Delta$, which implies \eqref{eq:equiv2}.
        \item Now suppose that \eqref{eq:equiv1} fails.  This means that there exists $V \in \calV$ such that $\Gamma(V_{X\widetilde{X}Y}) \le \Delta$ and $R \ge I_V(X;\widetilde{X})$, which implies that \eqref{eq:equiv2} fails.
    \end{itemize}
    Finally, we note that the right-hand side of \eqref{eq:equiv2} is precisely $E_{\mathrm{RGV}}(R,P,W,q,d, \Delta)|_{d=-I(X;\widetilde{X})}$ (see (\ref{eq: E_ex dfn})), and we recall that $\Delta$ equals the achievable exponent for $d(P_{X\widetilde{X}})=\beta_{R,W,q}(P_{X\widetilde{X}})$.  Thus, \eqref{eq:equiv2} states that given $R$, this exponent can be made arbitrarily close to $E_q(R,P,W)$ in \eqref{eq: E_ex dfn C and K}.  Since the latter is optimal by Proposition \ref{th: alpha optimality}, the proof is complete.
\end{proof}

\subsection{Bhattacharyya and Chernoff Distances} \label{sc:Bhat}

Here we show that an additive distance function with per-letter distance
\begin{equation}
    d_s(x,x') = -\log\sum_{y} W(y|x)\bigg(\frac{\e^{q(x',y)}}{\e^{q(x,y)}}\bigg)^s, \label{eq:chernoff_dist}
\end{equation}
for suitably-chosen $s > 0$ also recovers the maximum of the random coding and expurgated exponents.  We call this the {\em Chernoff distance}, because it is closely related to the Chernoff bound for bounding a probability of the event of the form $\{q^{}(\bX',\by) \ge q(\bx,\by)\}$.  In the case of ML decoding $q(x,y) = \log W(y|x)$, we choose $s = \frac{1}{2}$, and hence $d_s$ reduces to the Bhattacharrya distance, which is symmetric.  For general decoding metrics, we may require $s \ne \frac{1}{2}$, and thus $d_s$ is not symmetric; however, the RGV exponent is still achievable according to Corollary \ref{cor:asymm}.

We note that since we are considering bounded metrics, the distance $d_s(x,x')$ is also bounded, in accordance with Definition \ref{def:distance}.  However, this may rule out certain choices such as $q(x,y) = \log W(y|x)$ for channels with zero-probability transitions, in which we wish to assign the value $q(x,y) = -\infty$ when $W(y|x) = 0$. 

We will show that the additive distance $d_s$ recovers both the random coding and expurgated exponents for mismatched decoding \cite{CsiszarKorner81graph,ScarlettPengMerhavMartinezGuilleniFabregas_mismatch_2014_IT}.  This implies the {\em near-optimality} of $d_s$, in the sense that no examples are known for which $E_q(R,P,W)$ is strictly higher than the maximum of the random-coding and expurgated exponents.

Recovering the (ensemble-tight) random coding exponent is immediate: By setting $\Delta$ equal its maximum possible value, the rate condition in \eqref{eq: frist eq} becomes trivial, and we can lower bound the exponent in \eqref{eq: E_ex dfn} by dropping the constraint $d(P_{X\widetilde{X}}) \ge \Delta$ and writing $I(\widetilde{X};Y,X) \ge I(\widetilde{X};Y)$.  The resulting exponent matches that of \cite{CsiszarKorner81,ScarlettMartinezGuilleniFabregas_mismatch_2014_IT}. Alternatively, setting $r=0$ in \eqref{eq:E_dual} gives the same exponent in the dual form.

Recovering the expurgated exponent is more difficult; we do this using the dual form in Theorem \ref{th:Lagrange}.  Setting $\rho=1$ in (\ref{eq:E_dual}), and letting $s$ coincide with the choice in \eqref{eq:chernoff_dist}, we obtain
\begin{flalign}
    &E_{\mathrm{RGV}}(R,P,W,q,d, \Delta)\notag\\ 
    & \ge-\sum_{x}P(x)\log\sum_{x'}P(x')\sum_{y}W(y|x)\nonumber\\
    & \quad \quad \times \Big(\frac{\e^{q(x',y)}}{\e^{q(x,y)}}\Big)^s \frac{\e^{a(x')}}{\e^{a(x)}}e^{r(d(x,x')-\Delta)}- R\\
    & =-\sum_{x}P(x)\log\sum_{x'}P(x')\e^{-d_{s}(x,x')}\frac{\e^{a(x')}}{\e^{a(x)}}\e^{r(d(x,x')-\Delta)}- R.
\end{flalign}
Setting $d=d_{s}$ and $r=\frac{\rho'}{1+\rho'}$ for some $\rho' \ge 0$
gives
\begin{flalign}
    &E_{\mathrm{RGV}}(R,P,W,q,d, \Delta) \nonumber \\
    &\ge-\sum_{x}P(x)\log\sum_{x'}P(x')\e^{-\frac{d_{s}(x,x')}{1+\rho'}}\frac{\e^{a(x')}}{\e^{a(x)}}+\Delta\frac{\rho'}{1+\rho'}- R.\label{eq:Bhat3}
\end{flalign}
Then, choosing
\begin{flalign}
    &\Delta=
    -(1+\rho')\bigg(\sum_{x}P(x)\log\bigg[ \sum_{x'}P(x')\e^{-\frac{d_{s}(x,x')}{1+\rho'}}\frac{\e^{a(x')}}{\e^{a(x)}}\bigg]\nonumber\\
    &\hspace*{5.9cm} +R+2\delta\bigg),\label{Bhat4}
\end{flalign}
we obtain from (\ref{eq:Bhat3}) that
\begin{align}
    &E_{\mathrm{RGV}}(R,P,W,q,d, \Delta) \nonumber \\
     & \ge-\sum_{x}P(x)\log\sum_{x'}P(x')\e^{-\frac{d_{s}(x,x')}{1+\rho'}}\frac{\e^{a(x')}}{\e^{a(x)}} \nonumber \\
        & -\rho'\bigg(\sum_{x}P(x)\log\sum_{x'}P(x')\e^{-\frac{d_{s}(x,x')}{1+\rho'}}\frac{\e^{a(x')}}{\e^{a(x)}}+R+2\delta\bigg)-R\\
    & =-(1+\rho')\bigg(\sum_{x}P(x)\log\sum_{x'}P(x')\e^{-\frac{d_{s}(x,x')}{1+\rho'}}\frac{\e^{a(x')}}{\e^{a(x)}}\bigg)\nonumber\\
    &\qquad -(1+\rho'+2\delta\rho')R.
\end{align}
Upon taking $\delta \to 0$ and optimizing over $\rho' \ge 0$, $s \ge 0$, and $a(\cdot)$, this exponent is identical to the dual form for the mismatched decoding expurgated exponent given in \cite{ScarlettPengMerhavMartinezGuilleniFabregas_mismatch_2014_IT}, which is known to be equivalent to the primal form given in \cite{CsiszarKorner81graph}.

We also need to check that
the choice of $\Delta$ in \eqref{Bhat4} complies with the rate condition in \eqref{eq:R_dual}. We choose the same $a(\cdot)$
as in the exponent, but a value different $r$ (note that the two need not be the same). We simplify the condition as follows:
\begin{flalign}
    R & \le-\sum_{x}P(x)\log\sum_{x'}P(x')\e^{a(x')-\phi_{a}}\e^{-r(d_{s}(x,x')-\Delta)}-2\delta\\
    & =-\sum_{x}P(x)\log\sum_{x'}P(x')\e^{a(x')-\phi_{a}}\e^{-rd_{s}(x,x')}-r\Delta-2\delta\\
    & =-\sum_{x}P(x)\log\sum_{x'}P(x')\e^{a(x')-\phi_{a}}\e^{-rd_{s}(x,x')} \nonumber \\
        &\quad + r(1+\rho')\bigg(\sum_{x}P(x)\log\bigg[\sum_{x'}P(x')\e^{-\frac{d_{s}(x,x')}{1+\rho'}}\frac{\e^{a(x')}}{\e^{a(x)}}\bigg] \nonumber\\
        &\hspace*{4.8cm} +R+2\delta\bigg)-2\delta, \label{eq:final_R_dual}
\end{flalign}
where we have substituted \eqref{Bhat4}.

By setting $r=\frac{1}{1+\rho'}$ and noting that $-\sum_{x}P(x)\log\sum_{x'}P(x')\e^{a(x')-\phi_{a}}\e^{-rd_{s}(x,x')}$ is identical to $-\sum_{x}P(x)\log\sum_{x'}P(x')\frac{\e^{a(x')}}{\e^{a(x)}}\e^{-rd_{s}(x,x')}$ (by expanding the logarithms and using $\phi_a = \sum_{x}P(x) a(x)$), we observe that \eqref{eq:final_R_dual} reduces to $R\le R$, which is
trivially satisfied.

\section{Discussion and Conclusion}\label{sc: Discussion}

In this paper, we introduced a sequential random scheme based on randomizing a generalized form of Gilbert-Varshamov codes with a general distance function. 
This ensemble ensures that the codewords are sufficiently separated in the input space, and simultaneously achieves both the random coding and expurgated exponents.  
We proved that the RGV exponent is ensemble-tight for any additive decoding metric, and to our knowledge, this is the first such result for any construction achieving the expurgated exponent.  In addition, we provided dual-domain expressions, along with a direct derivation that extends beyond the finite-alphabet setting, and we presented choices of the distance function that attain the best possible exponent.

\appendix

\subsection{Proof of Lemma \ref{lm: pairwise distribution of codewords}}
\label{app:pairwise lemma}

In the following, products of the form $\prod_{i \ne \{k,m\}}$ are a shorthand for $\prod_{i \in \{1,\dotsc,m\} \backslash \{k,m\}}$.  In addition, for the special case of $m=k+1$, any summations over $\bx_{k+1}^{m-1}$ are void, and any terms of the form $\Pr(\bx_{k+1}^{m-1}|\bx_1^k)$ should be omitted (i.e., replaced by $1$)
Since by assumption $k < m$, we have
\begin{flalign}
&\Pr(\bx_k,\bx_m) \nonumber  \\
& = \sum_{\bx_1^{k-1}, \bx_{k+1}^{m-1}} \Pr(\bx_1^{k-1}) \Pr(\bx_k|\bx_1^{k-1})  \nonumber \\
& \qquad \qquad \times \Pr(\bx_{k+1}^{m-1}|\bx_1^k) \Pr(\bx_{m}|\bx_{1}^{m-1}) \\
& = \sum_{\bx_1^{k-1}, \bx_{k+1}^{m-1}} \Pr(\bx_1^{k-1}) \Pr(\bx_{k+1}^{m-1}|\bx_1^k) \nonumber \\
& \qquad \times \frac{\prod_{i=1}^{k-1} \indicator\{d(\bx_k,\bx_i)>\Delta\}}{|\calT(P_n,\bx_1^{k-1})|} \frac{\prod_{i=1}^{m-1} \indicator\{d(\bx_{m},\bx_i)>\Delta\}}{|\calT(P_n,\bx_1^{m-1})|} \label{eq:lem7eq0}\\
&\geq \frac{\indicator\{d(\bx_k,\bx_{m})>\Delta\}}{|\calT(P_n)|^2} \sum_{\bx_1^{k-1}, \bx_{k+1}^{m-1}} \Pr(\bx_1^{k-1}) \Pr(\bx_{k+1}^{m-1}|\bx_1^k) \nonumber \\
& \qquad \quad\times \prod_{i=1}^{k-1} \indicator\{d(\bx_k,\bx_i)>\Delta\} \prod_{i \notin \{k,m\}} \indicator\{d(\bx_{m},\bx_i)>\Delta\}\label{eq:lem71stineq}\\
&= \frac{\indicator\{d(\bx_k,\bx_{m})>\Delta\}}{|\calT(P_n)|^2} \sum_{\bx_1^{k-1}, \bx_{k+1}^{m-1}} \prod_{i \notin \{k,m\}} \Pr(\bx_i|\bx_1^{i-1}) \nonumber \\
& \qquad \times\indicator\{d(\bx_k,\bx_i)>\Delta\} \indicator\{d(\bx_{m},\bx_i)>\Delta\} \label{eq:lem72ndineq}
\end{flalign}
where \eqref{eq:lem7eq0} follows by noting that the two fractions appearing are precisely $\Pr(\bx_k|\bx_1^{k-1})$ and $\Pr(\bx_{m}|\bx_1^{m-1})$, \eqref{eq:lem71stineq} follows from Lemma \ref{lem:Ti_bounds}, and \eqref{eq:lem72ndineq} writes $\Pr(\bx_1^{k-1}) \Pr(\bx_{k+1}^{m-1}|\bx_1^k)$ recursively, as well as extending $\prod_{i=1}^{k-1} \indicator\{d(\bx_k,\bx_i)>\Delta\}$ to $\prod_{i \notin \{k,m\}} \indicator\{d(\bx_k,\bx_i)>\Delta\}$ since the term $\Pr(\bx_1^{k-1}) \Pr(\bx_{k+1}^{m-1}|\bx_1^k)$ is zero whenever $d(\bx_k,\bx_i) \le \Delta$ for some $k < i < m$.

We now apply a recursive procedure to the summation in \eqref{eq:lem72ndineq}.  Letting $\psi_i(\bx_i,\bx_1^{i-1},\bx_k,\bx_m)$ denote the argument to the product therein, we have
\begin{multline}
\sum_{\bx_1^{k-1}, \bx_{k+1}^{m-1}} \prod_{i \notin \{k,m\}} \psi_i(\bx_i,\bx_1^{i-1},\bx_k,\bx_m) \\ 
= \bigg(\sum_{\bx_1^{k-1}, \bx_{k+1}^{m-2}}\prod_{i \notin \{k,m,m-1\}} \psi_i(\bx_i,\bx_1^{i-1},\bx_k,\bx_m)\bigg) \\
\times \sum_{\bx_{m-1}} \psi_{m-1}(\bx_{m-1},\bx_1^{m-2},\bx_k,\bx_m). \label{eq:recursion}
\end{multline}
The summation over $\bx_{m-1}$ can be expanded as follows:
{\allowdisplaybreaks 
\begin{align}
& \sum_{\bx_{m-1}} \psi_{m-1}(\bx_{m-1},\bx_1^{m-2},\bx_k,\bx_m) \nonumber \\
&\quad = \sum_{\bx_{m-1}} \frac{ \indicator\{ \bx_{m-1}\in \calT(P_n,\bx_1^{m-2}) \} }{ |\calT(P_n,\bx_1^{m-2})| } \indicator\{d(\bx_k,\bx_{m-1})>\Delta\} \nonumber\\
&\qquad \qquad \times \indicator\{d(\bx_{m},\bx_{m-1})>\Delta\} \label{eq:single_sum_1} \\
&\quad = \frac{ |\calT(P_n,\bx_1^{m-2},\bx_k,\bx_{m})| }{ |\calT(P_n,\bx_1^{m-2})| } \label{eq:single_sum_2} \\
&\quad \ge \frac{ |\calT(P_n,\bx_1^{m-2})| - 2\,\vol_{\bx}(\Delta) }{ |\calT(P_n,\bx_1^{m-2})| } \label{eq:single_sum_3}\\
&\quad = 1 - \frac{ 2\,\vol_{\bx}(\Delta) }{ |\calT(P_n,\bx_1^{k-2})| } \\
&\quad \ge 1 - \frac{ 2e^{-n(R_n+\delta)} }{ 1-e^{-n\delta} } \label{eq:single_sum_4} \\
&\quad  = 1 - 2 \delta_ne^{-nR_n}, \label{eq:single_sum_5}
\end{align}}
\noindent where \eqref{eq:single_sum_2} follows since the three indicator functions are simultaneously equal to one if and only if $\bx_{m-1} \in \calT(P_n,\bx_1^{m-2},\bx_k,\bx_{m})$, \eqref{eq:single_sum_3} follows since the only sequences that can be in $\calT(P_n,\bx_1^{m-2})$ but not $\calT(P_n,\bx_1^{m-2},\bx_k,\bx_{m})$ are those in the $d$-balls centered as $\bx_k$ and $\bx_{m}$ (recall also that $\vol_{\bx}$ does not depend on $\bx$), and \eqref{eq:single_sum_4} follows from the volume upper bound and the set cardinality lower bound Lemma \ref{lem:Ti_bounds}, and \eqref{eq:single_sum_5} applies the definition of $\delta_n$ in \eqref{eq: delta_n dfn }.

Applying the above procedure recursively to the indices $m-2$, $m-3$, and so on in \eqref{eq:recursion} (skipping index $k$), and substituting into \eqref{eq:lem72ndineq}, we obtain
\begin{flalign}
\Pr(\bx_k,\bx_m)&\geq \frac{\indicator\{d(\bx_k,\bx_{m})>\Delta\}}{|\calT(P_n)|^2} \bigg( 1 - \frac{2\delta_n}{\e^{nR_n}}\bigg)^{\e^{nR_n}}\label{eq:lem72ineq3}\\
&\geq \frac{\indicator\{d(\bx_k,\bx_{m})>\Delta\}}{|\calT(P_n)|^2} (1-4\delta_n^2)\e^{-2\delta_n}\label{eq:lem72ineq4}
\end{flalign}
where \eqref{eq:lem72ineq3} also applies $m-2 \le e^{nR_n}$ in the exponent, and \eqref{eq:lem72ineq4} follows from the standard inequality $\big(1-\frac{\alpha}{N}\big)^N \geq  e^{-\alpha} \big(1-\frac{\alpha^2}{N}\big)$. This establishes the desired lower bound. 

The upper bound in \eqref{eq: pairwise distribution of codewords k m} simply follows by applying Lemma \ref{lem:Ti_bounds} to \eqref{eq:lem7eq0}, and upper bounding the indicator functions by one.

\subsection{Proof of Lemma \ref{lm: triplets-wise distribution of codewords}}
\label{app:triplets distribution lemma}

Recall the abbreviation in \eqref{eq: indicator triple} (which we use with $k$ in place of $m$). Recalling the assumption $i < j < k$, we have
\begin{flalign}
&\Pr(\bx_i,\bx_j,\bx_k) \nonumber \\
& = \sum_{\bx_1^{i-1}, \bx_{i+1}^{j-1}, \bx_{j+1}^{k-1}} \Pr(\bx_1^{i-1}) \Pr(\bx_i|\bx_1^{i-1}) \Pr(\bx_{i+1}^{j-1}|\bx_1^i) \nonumber \\
& \qquad \times\Pr(\bx_{j}|\bx_{1}^{j-1}) 
\Pr(\bx_{j+1}^{k-1}|\bx_1^j) \Pr(\bx_{k}|\bx_{1}^{k-1}) \label{eq:lem7eq0 aejh0}\\
& = \sum_{\bx_1^{i-1}, \bx_{i+1}^{j-1}, \bx_{j+1}^{k-1}} \Pr(\bx_1^{i-1}) \Pr(\bx_{i+1}^{j-1}|\bx_1^i)
\Pr(\bx_{j+1}^{k-1}|\bx_1^j)\nonumber\\
    &\qquad\times \frac{\prod_{r=1}^{i-1} \indicator\{d(\bx_i,\bx_r)>\Delta\}}{|\calT(P_n,\bx_1^{i-1})|} \frac{\prod_{s=1}^{j-1} \indicator\{d(\bx_{j},\bx_s)>\Delta\}}{|\calT(P_n,\bx_1^{j-1})|}
\nonumber \\
& \qquad \times\frac{\prod_{t=1}^{m-1} \indicator\{d(\bx_{m},\bx_t)>\Delta\}}{|\calT(P_n,\bx_1^{k-1})|}
  \label{eq:lem7eq0 aejh}\\
&\leq \frac{\calI_{d,\Delta}(\bx_i,\bx_j,\bx_k)}{(1-e^{-n\delta})^3|\calT(P_n)|^3} \sum_{\bx_1^{i-1}, \bx_{i+1}^{j-1}, \bx_{j+1}^{k-1}} \Pr(\bx_1^{i-1}) \nonumber \\
& \qquad \times\Pr(\bx_{i+1}^{j-1}|\bx_1^i) \Pr(\bx_{j+1}^{k-1}|\bx_1^j) \label{eq:lem1eq0 aejh liuegafuyg}
\end{flalign}
\begin{flalign}
&= \frac{\calI_{d,\Delta}(\bx_i,\bx_j,\bx_k)}{(1-e^{-n\delta})^3|\calT(P_n)|^3} \sum_{\bx_1^{i-1}}\Pr(\bx_1^{i-1}) \nonumber \\
& \qquad \times\sum_{ \bx_{i+1}^{j-1}}\Pr(\bx_{i+1}^{j-1}|\bx_1^i) \sum_{ \bx_{j+1}^{k-1}} \Pr(\bx_{j+1}^{k-1}|\bx_1^j) \\
&= \frac{\calI_{d,\Delta}(\bx_i,\bx_j,\bx_k)}{(1-e^{-n\delta})^3|\calT(P_n)|^3} ,
\label{eq:lem72ndineq euvgf}
\end{flalign}
where \eqref{eq:lem7eq0 aejh0} substitutes the conditional codeword distributions given all previous codewords, and \eqref{eq:lem1eq0 aejh liuegafuyg} uses Lemma \ref{lem:Ti_bounds}.

\subsection{Proof of Lemma \ref{lm: marginal distribution Xm lemma}} \label{ap: Proof of Uniformity}

Let $\pi$ be a permutation of the indices $[1,\dotsc,n]$, and let $\pi(\bx)$ be the outcome of applying the permutation $\pi$ to the sequence $\bx$. 
By the definition of the generalized RGV construction (in particular, the fact that the codewords are drawn uniformly and $d$ is type-dependent), we have
\begin{multline}
\Pr\left(\bX_1=\bx_1,\bX_2=\bx_2,\dotsc,\bX_m=\bx_m\right) \\
 =\Pr\left(\bX_1=\pi(\bx_1),\bX_2=\pi(\bx_2),\dotsc,\bX_m=\pi(\bx_m)\right).
\end{multline}
We now consider summing both sides over all sequences $(\bx_1,\dotsc,\bx_{m-1})$ that are admissible in the sense of meeting the requirement $d(\bx_i,\bx_j)>\Delta$ for all $i,j\in\{1,...,m\}$. 
 Clearly such a summation yields $\Pr\left(\bX_m=\bx_m\right)$ on the left-hand side.  Moreover, for each such $(\bx_1,\dotsc,\bx_{m-1})$, the type-dependent nature of $d$ implies that $(\pi(\bx_1),\dotsc,\pi(\bx_{m-1}))$ and $(\pi^{-1}(\bx_1),\dotsc,\pi^{-1}(\bx_{m-1}))$ are also admissible.  As a result, we are also summing the right-hand side over all admissible sequences, yielding 
\begin{flalign}
\Pr\left(\bX_m=\bx_m\right)=\Pr\left(\bX_m=\pi(\bx_m)\right),
\end{flalign}
which implies that $\bX_m$ is distributed uniformly over $\calT(P_n)$. 

\subsection{Proof of Lemma \ref{lem:continuity}} \label{sec:pf_continuity}

The RGV exponent, defined in \eqref{eq: E_ex dfn}, is a minimization over joint distributions $V_{X\widetilde{X}Y}$ within the constraint set $\calT_{d,q,P}(\Delta)$ given in \eqref{eq: cal T alpha definition}.  

Let $V^*_{X\widetilde{X}Y}$ denote the minimizer subject to $\calT_{d,q,P}(\Delta)$, and let $V^*_{X\widetilde{X}Y,n}$ denote the minimizer subject to $\calT_{d,q,P_n}(\Delta)$.  Since the space of probability distributions is compact, any infinite subsequence of $V^*_{X\widetilde{X}Y,n}$ must have a further subsequence converging to some $V^*_{X\widetilde{X}Y,\infty}$. Moreover, since $d$ and $q$ are continuous and $V^*_{X\widetilde{X}Y,n} \in \calT_{d,q,P_n}(\Delta)$ with $P_n \to P$, we must have $V^*_{X\widetilde{X}Y,\infty} \in \calT_{d,q,P}(\Delta)$, from which \eqref{eq:cont_LB} follows.

\subsection{Proof of Lemma \ref{lm: empirical simplex lemma}}\label{ap: empirical simplex lemma appendix}

In this appendix, we make use of the following notation, also used in Section \ref{sec:non_univ}:
\beq\label{eq: gamma function definition2}
    \Gamma(V_{X\tilde X Y}) \triangleq D(V_{Y|X}\|W|V_X)+\big|I_{}(\widetilde{X};Y,X)-R\big|_+.
\eeq
We observe that the exponent on the right-hand side of (\ref{eq: standard method of type 2}) can be rewritten as 
\begin{flalign}
\min_{\substack{V_{X\widetilde{X}}\in\mathcal{P}_{n}(\calX^2)\,:\,\\ V_{X}=V_{\widetilde{X}}=P_{n},\;\\
d(V_{X\widetilde{X}})\ge\Delta
}
}\min_{\substack{V_{Y|X\widetilde{X}}\in\mathcal{P}_{n}(\calY|V_{X\widetilde{X}})\,:\, \\ q(P_{n}\times V_{Y|\widetilde{X}})-q(P_{n}\times V_{Y|X})\ge0}} \Gamma(V_{X\widetilde{X}}\times V_{Y|X\widetilde{X}}),\label{eq:double_min 11}
\end{flalign}
where the notation $V_{Y|X\widetilde{X}}\in\mathcal{P}_{n}(\calY|V_{X\widetilde{X}})$ means that $V_{X\widetilde{X}}\times V_{Y|X\widetilde{X}}$ is a joint empirical distribution for sequences of length $n$.  Throughout the appendix, we will make use of the fact the minimizers must be such that
\begin{flalign}\label{eq: artificial constraint}
	W(y|x)=0 \implies V_{Y|X}(y|x) = 0, 
\end{flalign}
since otherwise the KL divergence in \eqref{eq: gamma function definition2} would be infinite.  Observe that within the space of joint distributions satisfying \eqref{eq: artificial constraint}, the function $\Gamma(\cdot)$ is continuous. 

We first show that the inner minimization can be approximated by a minimization over $V_{Y|X\widetilde{X}}\in\mathcal{P}(\calY|\calX^2)$, and then we show that the outer minimization can be approximated by a minimization over $V_{X\widetilde{X}}\in\mathcal{P}(\calX^2)$.

\textbf{Inner minimization. }
Define $\Psi(V_{X\widetilde{X}}\times V_{Y|X\widetilde{X}})=q(P_{n}\times V_{Y|\widetilde{X}})-q(P_{n}\times V_{Y|X})$,
so that the constraint in \eqref{eq:double_min 11} is given by $\Psi(V_{X\widetilde{X}}\times V_{Y|X\widetilde{X}})\ge0$. For any $V_{X\widetilde{X}}\in \mathcal{P}_{n}(\calX^2)$, 
we need to show that 
the inner minimization in (\ref{eq:double_min 11})
can be expanded from $\mathcal{P}_{n}(\calY|V_{X\widetilde{X}})$ to $\mathcal{P}(\calY|\calX^2)$. Specifically, we wish to show that for any $\epsilon>0$, it holds for sufficiently large $n$ that 
\begin{align}
	&\min_{V_{Y|X\widetilde{X}}\in\mathcal{P}_{n}(\calY|V_{X\widetilde{X}})\,:\,\Psi(V_{X\widetilde{X}}\times V_{Y|X\widetilde{X}})\ge0}\Gamma(V_{X\widetilde{X}}\times V_{Y|X\widetilde{X}})\nonumber \\
	& \leq \min_{V_{Y|X\widetilde{X}}\in\mathcal{P}(\calY|\calX^2)\,:\,\Psi(V_{X\widetilde{X}}\times V_{Y|X\widetilde{X}})\ge0} \Gamma(V_{X\widetilde{X}}\times V_{Y|X\widetilde{X}})+\epsilon.
	\label{eq:double_min}
\end{align}
Since we are considering additive decoding metrics, i.e., $q(P_{XY}) = \EE_P[q(X,Y)]$, we have
\begin{flalign}
	&\Psi(V_{X\widetilde{X}}\times V_{Y|X\widetilde{X}}) \nonumber \\
	&=\sum_{x,\overline{x},y} V_{X\widetilde{X}}(x,\overline{x})V_{Y|X\widetilde{X}}(y|x,\overline{x}) \cdot [q(\overline{x},y)-q(x,y)]. \label{eq:Psi_expansion}
\end{flalign}
To prove (\ref{eq:double_min}), fix any $\widetilde{V}_{Y|X\widetilde{X}}\in\mathcal{P}(\calY|\calX^2)$ with
$\Psi(V_{X\widetilde{X}}\times
\widetilde{V}_{Y|X\widetilde{X}})\ge0$, and let $V_{Y|X\widetilde{X}}^{(n)}$
be the quantized version of $\widetilde{V}_{Y|X\widetilde{X}}$ that rounds up for the highest values of $q(\overline{x},y)-q(x,y)$, and rounds down for the smallest values:
\begin{flalign}
&V^{(n)}_{Y|X\widetilde{X}}(y|x,\overline{x}) \nonumber\\
&=\left\{\begin{array}{l} \frac{1}{nV_{X\widetilde{X}}(x,\overline{x})}\lceil n\cdot V_{X\widetilde{X}}(x,\overline{x})\cdot  \widetilde{V}_{Y|X\widetilde{X}}(y|x,\overline{x})\rceil \\\qquad \qquad \mbox{ if } q(\overline{x},y)-q(x,y) > c_{x\overline{x}}
\\
\frac{1}{nV_{X\widetilde{X}}(x,\overline{x})}\lfloor n\cdot V_{X\widetilde{X}}(x,\overline{x})\cdot \widetilde{V}_{Y|X\widetilde{X}}(y|x,\overline{x})\rfloor \\\qquad \qquad \mbox{  if }  q(\overline{x},y)-q(x,y) < c_{x\overline{x}},
  \end{array}\right.
\end{flalign}
where for each $(x,\overline{x})$, we choose $c_{x\overline{x}}$ (as well as rounding the entries with $q(\overline{x},y)-q(x,y) = c_{x\overline{x}}$ up or down as needed) in such a way that the entries of $V^{(n)}_{Y|X\widetilde{X}}(y|x,\overline{x})$ sum to one. 

By this construction and the fact that $\Psi(V_{X\widetilde{X}}\times V_{Y|X\widetilde{X}})$ is a positive linear combination of the values $q(\overline{x},y)-q(x,y)$ ({\em cf.}, \eqref{eq:Psi_expansion}),  we have
\begin{gather}
\Psi(V_{X\widetilde{X}}\times V^{(n)}_{Y|X\widetilde{X}})\geq \Psi(V_{X\widetilde{X}}\times \widetilde{V}_{Y|X\widetilde{X}} )\label{eq: q gap 1} \end{gather}
and 
\begin{multline}
\hspace*{-1ex}\sum_{x,\widetilde{x},y} \big| V_{X\widetilde{X}}(x,\widetilde{x}) V^{(n)}_{Y|X\widetilde{X}}(y|x,\widetilde{x})-V_{X\widetilde{X}}(x,\widetilde{x}) \widetilde{V}_{Y|X\widetilde{X}}(y|x,\widetilde{x})\big| \\
\leq \frac{|\calX|^2|\calY|}{n}\label{eq: l1 distance small}.
\end{multline}
In particular, \eqref{eq: q gap 1} immediately implies that the required constraint $\Psi(V_{X\widetilde{X}}\times V^{(n)}_{Y|X\widetilde{X}})\ge0$ is satisfied. Moreover, (\ref{eq: l1 distance small}) implies that $V_{X\widetilde{X}}\times V_{Y|X\widetilde{X}}^{(n)}$ is $O\big(\frac{1}{n}\big)$-close
to $V_{X\widetilde{X}}\times \widetilde{V}_{Y|X\widetilde{X}}$ (in the $\ell_1$ sense), and hence $\Gamma(V_{X\widetilde{X}}\times \widetilde{V}_{Y|X\widetilde{X}}) - \Gamma(V_{X\widetilde{X}}\times V_{Y|X\widetilde{X}}^{(n)}) \to 0$ by the continuity of $\Gamma(\cdot)$. This proves the part of the approximation of the inner minimization, i.e., (\ref{eq:double_min}). 

\textbf{Outer minimization.} 
Having proved (\ref{eq:double_min}), the double minimization (\ref{eq:double_min 11}) is upper bounded by the following double minimization:
\begin{flalign}
	 & \min_{\substack{V_{X\widetilde{X}}\in\mathcal{P}_{n}(\calX^2)\,:\\V_{X}=V_{\widetilde{X}}=P_{n},\;\\d(V_{X\widetilde{X}})\ge\Delta}} \min_{\substack{V_{Y|X\widetilde{X}}\in\mathcal{P}(\calY|\calX^2)\,:\, \\ q(P_{n}\times V_{Y|\widetilde{X}})-q(P_{n}\times V_{Y|X})\ge0}}\Gamma(V_{X\widetilde{X}}\times V_{Y|X\widetilde{X}})
	 . \label{eq:single_min_0}
\end{flalign}
Consider the expression in \eqref{eq:single_min_0} with $\mathcal{P}_n(\calX^2)$ replaced by $\mathcal{P}(\calX^2)$ and $P_n$ replaced by $P$:
\begin{flalign}
& \min_{\substack{V_{X\widetilde{X}}\in\mathcal{P}(\calX^2)\,:\\V_{X}=V_{\widetilde{X}}=P,\;\\
d(V_{X\widetilde{X}})\ge\Delta
}
}\min_{\substack{V_{Y|X\widetilde{X}}\in\mathcal{P}(\calY|\calX^2)\,:\, \\ q(P\times V_{Y|\widetilde{X}})-q(P\times V_{Y|X})\ge0
}
}\Gamma(V_{X\widetilde{X}}\times V_{Y|X\widetilde{X}}).\label{eq: three constraints}
\end{flalign}
Given the minimizer $V_{X\widetilde{X}}^*\in\mathcal{P}(\calX^2)$ with $V_{X}^*=V_{\widetilde{X}}^*=P$, let $V^*_{X\widetilde{X},n}$ be the closest joint type (e.g., in the $\ell_\infty$ sense) that satisfies $V^*_{X}=V^*_{\widetilde{X}}=P_{n}$. It follows that $V^*_{X\widetilde{X},n}(x,\overline{x}) - V_{X\widetilde{X}}^*(x,\overline{x}) \to 0$. 

Let $V_{Y|X\widetilde{X}}^*$ denote the minimizer in \eqref{eq: three constraints}, and define
\begin{equation}
	V_{Y|X\widetilde{X},n}^{\max} = \argmax_{V_{Y|X\widetilde{X}}\in\mathcal{P}_n(\calY|\calX^2)} \Psi(V_{X\widetilde{X},n}\times V_{Y|X\widetilde{X}}). \label{eq:Vmax}
\end{equation}
We claim that there exists a vanishing sequence $\epsilon_n$ such that 
\begin{flalign}
	&(1-\epsilon_n)\Psi(V_{X\widetilde{X},n}^*\times V_{Y|X\widetilde{X}}^*)
	+\epsilon_n
	\Psi(V_{X\widetilde{X},n}\times V_{Y|X\widetilde{X},n}^{\max})\geq 0. \label{eq:eps_n_claim}
\end{flalign}
To see this, note that since $\Psi(V_{X\widetilde{X}}^*\times V_{Y|X\widetilde{X}}^*) \ge 0$ by definition, we only need the second term in \eqref{eq:eps_n_claim} to be large enough to overcome the rounding from $V^*_{X\widetilde{X}}$ to $V^*_{X\widetilde{X},n}$.  If $\Psi(V_{X\widetilde{X},n}\times V_{Y|X\widetilde{X},n}^{\max}) > 0$, then this is possible by letting $\epsilon_n$ vanish sufficiently slowly.  On the other hand,  $\Psi(V_{X\widetilde{X},n}\times V_{Y|X\widetilde{X},n}^{\max}) < 0$ is impossible, since one could swap the roles of $X$ and $\widetilde{X}$ in \eqref{eq:Vmax} to produce a positive quantity.  The only remaining case is that $\Psi(V_{X\widetilde{X},n}\times V_{Y|X\widetilde{X}}) = 0$ for all $V_{Y|X\widetilde{X}}$, in which case \eqref{eq:eps_n_claim} is trivial.

Using \eqref{eq:eps_n_claim} and the continuity of $\Gamma$ (subject to \eqref{eq: artificial constraint}, which we have established to always hold), we deduce the following for any $\epsilon>0$ and sufficiently large $n$: 
\begin{flalign}
 & \min_{\substack{V_{X\widetilde{X}}\in\mathcal{P}(\calX^2)\,:\\V_{X}=V_{\widetilde{X}}=P,\;\\
d(V_{X\widetilde{X}})\ge\Delta
}
}\min_{\substack{V_{Y|X\widetilde{X}}\in\mathcal{P}(\calY|\calX^2)\,:\, \\ q(P\times V_{Y|\widetilde{X}})-q(P\times V_{Y|X})\ge0}}\Gamma(V_{X\widetilde{X}}\times V_{Y|X\widetilde{X}}) \label{eq:RGV_alt}\\
 & = \Gamma(V_{X\widetilde{X}}^*\times V_{Y|X\widetilde{X}}^*) 
\\
 & \geq\Gamma(V^*_{X\widetilde{X},n}\times V_{Y|X\widetilde{X}}^*)-\epsilon
\label{eq: continuity inequality A}\\
 & \geq\Gamma\left(V^*_{X\widetilde{X},n}\times \left[(1-\epsilon_n)V_{Y|X\widetilde{X}}^*+\epsilon_nV_{Y|X\widetilde{X},n}^{\max}\right]\right)-2\epsilon
\label{eq: continuity inequality B}\\
 & \geq\min_{\substack{V_{Y|X\widetilde{X}}\in\mathcal{P}_n(\calY|\calX^2)\,:\, \\ q(P_n\times V_{Y|\widetilde{X}})-q(P_n\times V_{Y|X})\ge0}}\Gamma\left(V^*_{X\widetilde{X},n}\times V_{Y|X\widetilde{X}}\right)-2\epsilon \label{eq: continuity inequality C} \\
 & \geq\min_{\substack{V_{X\widetilde{X}}\in\mathcal{P}(\calX^2)\,:\\V_{X}=V_{\widetilde{X}}=P_{n},\;\\d(V_{X\widetilde{X}})\ge\Delta-\epsilon}} \min_{\substack{V_{Y|X\widetilde{X}}\in\mathcal{P}_n(\calY|\calX^2)\,:\, \\ q(P_n\times V_{Y|\widetilde{X}})-q(P_n\times V_{Y|X})\ge0}}\Gamma\left(V_{X\widetilde{X}}\times V_{Y|X\widetilde{X}}\right)\nonumber\\
 &\hspace*{6.9cm}-2\epsilon \label{eq: continuity inequality D}
\end{flalign}
where both \eqref{eq: continuity inequality A} and \eqref{eq: continuity inequality B} follow from the continuity of $\Gamma(\cdot)$, \eqref{eq: continuity inequality C} follows since $(1-\epsilon_n)V_{Y|X\widetilde{X}}^*+\epsilon_nV_{Y|X\widetilde{X},n}^{\max}$ belongs to the constraint set in the minimization due to \eqref{eq:eps_n_claim}, and \eqref{eq: continuity inequality D} follows since $d(V_{X\widetilde{X}}^*)\geq \Delta \implies d(V_{X\widetilde{X},n}^*)\geq \Delta - \epsilon$ by the continuity of $d$.

Since $\Delta$ is arbitrary in the preceding steps, we may replace $\Delta$ by $\Delta + \epsilon$ in both \eqref{eq:RGV_alt} and \eqref{eq: continuity inequality D}.  Upon doing so, we obtain the RGV exponent with input distribution $P$ and parameter $\Delta + \epsilon$ on the left-hand side, while recovering the expression \eqref{eq:single_min_0} from the first step above on the right-hand side.  This completes the proof of Lemma \ref{lm: empirical simplex lemma}.

\subsection{Primal-dual Equivalence} \label{sc:Lagrange}

The primal-dual equivalence stated in Theorem \ref{th:Lagrange} follows in a near-identical manner to the mismatched random coding exponent \cite{ScarlettPengMerhavMartinezGuilleniFabregas_mismatch_2014_IT} (and to a lesser extent, the mismatched expurgated exponent \cite{ScarlettPengMerhavMartinezGuilleniFabregas_mismatch_2014_IT}), so we omit most of the details.  We first consider the exponent \eqref{eq:E_dual}, and then the rate condition \eqref{eq:R_dual}.  

{\bf Exponent expression.} The proof of equivalence for the exponent consists of three steps, interleaved with applications of the minimax theorem to swap the order of the primal and dual optimization variables:
\begin{enumerate}
    \item Let $P_{XY}$ be fixed, and consider the optimization problem
    \begin{equation}
        \min_{\substack{V_{X\widetilde{X}Y} \,:\, V_{XY} = P_{XY}, P_{\widetilde{X}} = P, \\  q(V_{\widetilde{X}Y}) \ge q(P_{XY}), d(P_{X\widetilde{X}}) \ge \Delta}} D\big( V_{X\widetilde{X}Y} \| P \times P_{XY} \big), \label{eq:opt_step1}
    \end{equation}
    where $ P \times P_{XY} $ denotes the joint distribution $P(x)P(\widetilde{x})P_{Y|X}(y|x)$.  This minimization arises from fixing the $(X,Y)$ marginals in \eqref{eq: E_ex dfn} and noting that all terms other than the mutual information $I(\widetilde{X};X,Y)$ are constant.  The mutual information is equivalent to the objective function in \eqref{eq:opt_step1}, due to the equality constraints.
    
    Applying Lagrange duality in the same way as the random coding setting \cite{ScarlettMartinezGuilleniFabregas_mismatch_2014_IT} (see also \cite[Appendix E]{ScarlettThesis}), we find that \eqref{eq:opt_step1} is equivalent to\footnote{We have $e^{q(x,y)}$ in place of $q(x,y)$ in \cite{ScarlettMartinezGuilleniFabregas_mismatch_2014_IT} because we are considering additive (rather than multiplicative) decoding rules.}
    \begin{flalign}
        &\sup_{s \ge 0, r \ge 0,a(\cdot)} -\sum_{x,y} P_{XY}(x,y)\nonumber\\
        &\times \log\frac{\sum_{x'} Q(x')e^{sq(x',y)} e^{a(x')} e^{r(d(x,x')-\Delta)} }{e^{sq(x,y)} e^{a(x)}}, \label{eq:opt_step1a}
    \end{flalign}
    where $s$, $r$, and $a(\cdot)$ are Lagrange multipliers corresponding to the metric constraint, distance constraint, and $\widetilde{X}$-marginal constraint.
    \item Let $g_{s,r,a}(x,y) = -\log\frac{\sum_{x'} Q(x')e^{sq(x',y)} e^{a(x')} e^{r(d(x,x')-\Delta)} }{e^{sq(x,y)} e^{a(x)}}$ be the function being averaged in \eqref{eq:opt_step1a}.  Based on the definition in \eqref{eq: E_ex dfn}, the previous step, and the minimax theorem, the RGV exponent is given by
    \begin{flalign}
         &\sup_{s \ge 0, r \ge 0,a(\cdot)} \min_{V_{XY}\,:\, P_X = P} D(V_{XY} \| P \times W)\nonumber\\
         &\qquad  + \big| \EE_V[g_{s,r,a}(X,Y)] - R \big|_+.
    \end{flalign}
    By applying $[z]_+ = \max_{\rho \in [0,1]} \rho z$ along with the minimax theorem, we find that this is equivalent to
    \begin{flalign}
         &\sup_{\rho\in[0,1], s \ge 0, r \ge 0,a(\cdot)} \min_{V_{XY}\,:\, P_X = P} D(V_{XY} \| P \times W) \nonumber\\
         &\qquad + \rho\big(\EE_V[g_{s,r,a}(X,Y)] - R\big). \label{eq:opt_step2a}
    \end{flalign}
    \item A minimization problem of the form \eqref{eq:opt_step2a} was already considered in \cite{ScarlettMartinezGuilleniFabregas_mismatch_2014_IT} (with a different choice of $g_{s,r,a}$), and it was shown that the minimization is equivalent to the expression
    \begin{equation}
        -\sum_{x} Q(x)\log\sum_{y} W(y|x)e^{\rho g_{s,r,a}(x,y)}.
    \end{equation}
    Substituting the definition of $g_{s,r,a}$ completes the proof.
\end{enumerate}

{\bf Rate condition expression.} The primal-dual equivalence for the rate condition can be proved using similar steps to those above; here we briefly discuss another way that it can be understood.

The primal expression \eqref{eq: frist eq} is of the same form as the so-called {\em LM rate} for mismatched decoding \cite{CsiszarKorner81graph,Hui83,CsiszarNarayan95}, with $\widetilde{X}$ playing the role of $Y$, and $d$ playing the role of the decoding metric.  Accordingly, the primal-dual equivalence is essentially a special case of that of the LM rate, which is well-established in the mismatched decoding literature \cite{MerhavKaplanLapidothShamai94,GantiLapidothTelatar2000,ScarlettThesis}.  

\subsection{Lower Bound for Marginal Distribution in Cost-Constrained Coding} \label{app:cost_lb}

Here we show that in the cost-constrained coding setting of Section \ref{sc:Lagrange}, each $\Pr(\bx_m)$ is lower bounded by $P_{\bX}(\bx_m)$ times a constant tending to one. Recall that the codeword distribution is of the form \eqref{eq:p_cost_m} with $1-e^{-n\delta} \le \mu_m(\bx_1^{m-1}) \le 1$ (see the properties following \eqref{eq:final0}).  We have
{\allowdisplaybreaks
\begin{flalign}
&\Pr(\bx_m) \nonumber \\
&= \sum_{\bx_1^{m-1}} \Pr(\bx_1^{m-1}) \Pr(\bx_{m}|\bx_{1}^{m-1}) \label{eq:dual_lb1} \\
& =  \sum_{\bx_1^{m-1}} \Pr(\bx_1^{m-1}) \frac{P_{\bX}(\bx_m)}{\mu_m(\bx_1^{m-1})}  \indicator\{ d(\bx_i,\bx_m) > \Delta, \, \forall i < m \} \label{eq:dual_lb2} \\
&\ge  P_{\bX}(\bx_m) \sum_{\bx_1^{m-1}} \Pr(\bx_1^{m-1}) \indicator\{ d(\bx_i,\bx_m) > \Delta, \, \forall i < m \} \label{eq:dual_lb3} \\
&\ge  P_{\bX}(\bx_m) \sum_{\bx_1^{m-1}} \prod_{i=1}^{m-1}\big( \Pr(\bx_i|\bx_1^{i-1}) \indicator\{ d(\bx_i,\bx_m) > \Delta \} \big), \label{eq:recursion_dual}
\end{flalign}}
where \eqref{eq:dual_lb2} substitutes \eqref{eq:p_cost_m}, \eqref{eq:dual_lb3} uses the fact that $\mu_m(\bx_1^{m-1}) \le 1$, and \eqref{eq:recursion_dual} is an expansion of $ \Pr(\bx_1^{m-1})$.

We now unravel the product one term at a time.  We start by writing
\begin{align}
&\sum_{\bx_1^{m-1}} \prod_{i=1}^{m-1}\big( \Pr(\bx_i|\bx_1^{i-1}) \indicator\{ d(\bx_i,\bx_m) > \Delta \}  \big) \nonumber  \\ 
    &= \sum_{\bx_1^{m-2}} \prod_{i=1}^{m-2}\big( \Pr(\bx_i|\bx_1^{i-1}) \indicator\{ d(\bx_i,\bx_m) > \Delta \}  \big) \nonumber \\
    &~~~\times \sum_{\bx_{m-1}} \Pr(\bx_{m-1}|\bx_1^{m-2}) \indicator\{ d(\bx_{m-1},\bx_m) > \Delta \}.
\end{align} 
Henceforth, let $\calD(\cdot)$ denote the set of possible codewords that are at a distance exceeding $\Delta$ from all codewords listed in the brackets.  Substituting the conditional codeword distribution for codeword $m-1$ gives
\begin{align}
&\sum_{\bx_{m-1}} \Pr(\bx_{m-1}|\bx_1^{m-2})  \indicator\{ d(\bx_{m-1},\bx_m) > \Delta \} \nonumber \\
&= \frac{1}{\mu_{m-1}(\bx_1^{m-2})} \sum_{\bx_{m-1}} P_{\bX}(\bx_{m-1}) \indicator\{\bx_{m-1} \in \calD( \bx_1^{m-2},\bx_m ) \} \\
&= \frac{ \Pr( \bX' \in \calD( \bx_1^{m-2},\bx_m ) ) }{  \Pr( \bX' \in \calD( \bx_1^{m-2} ) ) }, \label{eq:dual_lb6}
\end{align}
where $\bX' \sim P_{\bX}$, and the denominator in \eqref{eq:dual_lb6} follows since $\mu_{m-1}(\bx_1^{m-2}) = \Pr( \bX' \in \calD( \bx_1^{m-2} ))$ by definition.

Continuing, we write
\begin{align}
&\frac{ \Pr( \bX' \in \calD( \bx_1^{m-2},\bx_m ) ) }{  \Pr( \bX' \in \calD( \bx_1^{m-2} ) ) } \nonumber \\
&\ge \frac{ \Pr( \bX' \in \calD( \bx_1^{m-2} ) ) - \Pr( d(\bX',\bx_m) \le \Delta ) }{  \Pr( \bX' \in \calD( \bx_1^{m-2} ) ) } \label{eq:dual_lb7} \\
&= 1 - \frac{\Pr( d(\bX',\bx_m) \le \Delta )}{ \Pr( \bX' \in \calD( \bx_1^{m-2} ) ) } \label{eq:dual_lb8} \\
&\ge 1 - \frac{\Pr( d(\bX',\bx_m) \le \Delta )}{ 1 - e^{-n\delta} } \label{eq:dual_lb9} \\
&\ge 1 - \frac{e^{-n(R_n + \delta)}}{ 1 - e^{-n\delta} } \label{eq:dual_lb10} \\
&= 1 - \delta_n e^{-nR_n}, \label{eq:dual_lb11}
\end{align}
where \eqref{eq:dual_lb7} uses $\Pr(A \cap B) \ge \Pr(A) - \Pr(B^c)$, \eqref{eq:dual_lb8} applies $\mu_{m-1}(\bx_1^{m-2}) \ge 1-e^{-n\delta}$, \eqref{eq:dual_lb9} makes use of the bounds on $\Pr( d(\bX',\bx_m) \le \Delta )$ and $R_n$ in \eqref{eq:final0} and \eqref{eq:R_dual} respectively, and \eqref{eq:dual_lb11} uses the definition of $\delta_n$ in \eqref{eq: delta_n dfn }.

The recursion in \eqref{eq:recursion_dual} proceeds in the exact same way as the constant-composition case in Appendix \ref{app:pairwise lemma} (with a factor of $2$ removed), and we get
\begin{equation}
\Pr(\bx_m) \ge P_{\bX}(\bx_m) \cdot (1-\delta_n^2)e^{-\delta_n}.
\end{equation}

\section*{Acknowledgment}
We would like to thank the reviewers for their very helpful comments and in particular for the suggestion on how to significantly shorten the proof of the achievability part of Theorem \ref{th: main theorem sphere packing}.

\begin{IEEEbiographynophoto}{Anelia Somekh-Baruch} 
(S'01-M'03) received the B.Sc. degree from Tel-Aviv University, Tel-Aviv, Israel, in 1996 and the M.Sc. and Ph.D.\ degrees from the Technion?Israel Institute of Technology, Haifa, Israel, in 1999 and 2003, respectively, all in electrical engineering. During 2003?2004, she was with the Technion Electrical Engineering Department. During 2005?2008, she was a Visiting Research Associate at the Electrical Engineering
Department, Princeton University, Princeton, NJ. From 2008 to 2009 she was a researcher at the Electrical Engineering Department, Technion, and from 2009 she has been with the Bar-Ilan University School of Engineering,
Ramat-Gan, Israel. Her research interests include topics in information theory and communication theory. Dr. Somekh-Baruch received the Tel-Aviv University program for outstanding B.Sc. students scholarship, the Viterbi
scholarship, the Rothschild Foundation scholarship for postdoctoral studies, and the Marie Curie Outgoing International Fellowship.
\end{IEEEbiographynophoto}

 \begin{IEEEbiographynophoto}{Jonathan Scarlett}
    (S'14 -- M'15) received 
    the B.Eng. degree in electrical engineering and the B.Sci. degree in 
    computer science from the University of Melbourne, Australia. 
    From October 2011 to August 2014, he
    was a Ph.D. student in the Signal Processing and Communications Group
    at the University of Cambridge, United Kingdom. From September 2014 to
    September 2017, he was post-doctoral researcher with the Laboratory for
    Information and Inference Systems at the \'Ecole Polytechnique F\'ed\'erale
    de Lausanne, Switzerland. Since January 2018, he has been an assistant
    professor in the Department of Computer Science and Department of Mathematics,
    National University of Singapore. His research interests are in
    the areas of information theory, machine learning, signal processing, and
    high-dimensional statistics. He received the Singapore National Research Foundation (NRF) fellowship, and the NUS Early Career Research Award.
\end{IEEEbiographynophoto}

\begin{IEEEbiographynophoto}{Albert Guill\'en i F\`abregas}(S'01--M'05--SM'09) received the Telecommunication Engineering degree and the Electronics Engineering degree from
the Universitat Polit\`ecnica de Catalunya, and the Politecnico di Torino,
Torino, Italy, respectively, in 1999, and the Ph.D. degree in Communication
Systems from Ecole Polytechnique F\'ed\'erale de Lausanne (EPFL), Lausanne,
Switzerland, in 2004.

Since 2011 he has been an ICREA Research Professor at Universitat
Pompeu Fabra. He is also an Adjunct Researcher at the University of
Cambridge. He has held appointments at the New Jersey Institute of Technology, Telecom Italia, European Space Agency (ESA), Institut Eur\'ecom,
University of South Australia, University of Cambridge, as well as visiting
appointments at EPFL, \'Ecole Nationale des T\'el\'ecommunications (Paris),
Universitat Pompeu Fabra, University of South Australia, Centrum Wiskunde
\& Informatica and Texas A\&M University in Qatar. His research interests are
in the areas of information theory, coding theory and communication theory.

Dr.  Guill\'en i F\`abregas is a member of the Young Academy of Europe,
received both Starting and Consolidator Grants from the European Research
Council, the Young Authors Award of the 2004 European Signal Processing
Conference (EUSIPCO), the 2004 Nokia Best Doctoral Thesis Award from
the Spanish Institution of Telecommunications Engineers, and a pre-doctoral
Research Fellowship of the Spanish Ministry of Education to join ESA.
He is an Associate Editor of the {\sc IEEE Transactions on Information Theory}, an Editor of the Foundations and Trends in Communications and
Information Theory and was an Editor of the {\sc IEEE Transactions on Wireless Communications}.
\end{IEEEbiographynophoto}

\end{document}